\documentclass[conference, thm-restate]{IEEEtran}
\IEEEoverridecommandlockouts

\usepackage{cite}
\usepackage{acronym}
\usepackage{graphicx}
\usepackage[utf8]{inputenc}
\usepackage{csquotes}
\usepackage{caption,subcaption}
\usepackage{textcomp}
\usepackage{xcolor}

\def\BibTeX{{\rm B\kern-.05em{\sc i\kern-.025em b}\kern-.08em
    T\kern-.1667em\lower.7ex\hbox{E}\kern-.125emX}}

\usepackage{siunitx}
\usepackage{amsmath}
\usepackage{amsfonts}
\usepackage{graphicx}
\usepackage{upgreek}
\usepackage{amsthm}
\usepackage{thmtools}
\usepackage{thm-restate}
\usepackage{booktabs}
\usepackage{amssymb}
\usepackage{mathtools}
\usepackage{pdflscape}
\usepackage{xspace}
\usepackage[abbreviations]{foreign}
\usepackage{makecell}

\usepackage{mdframed} % or, "mdframed"
\usepackage{lipsum}

\newmdtheoremenv{theo}{Theorem}

\usepackage{xcolor}%to write text  in different colors

\usepackage{algorithm} 
\usepackage{algpseudocode}

\usepackage{hyperref}

\usepackage{multicol}

\usepackage{mathtools}

\usepackage{cite}

\usepackage{enumitem}
\setlist[enumerate]{itemsep=0mm}

\newtheorem{lemma}{Lemma}

\newtheorem{theorem}{Theorem}

\theoremstyle{definition}
    
    \newtheorem{definition}{Definition}

\theoremstyle{plain}

\newtheorem{remark}{Remark}

\DeclareMathOperator{\swap}{Y}

\DeclareMathOperator{\sample}{b}%sample size
\DeclareMathOperator{\e}{\mathbf{exp}}
\DeclareMathOperator{\eh}{\hat{e}}
\DeclareMathOperator{\ih}{\hat{i}}
\DeclareMathOperator{\jh}{\hat{j}}

\renewcommand{\Pr}{\mathbb P}

\usepackage{ifthen}
\newboolean{showcomments}
\setboolean{showcomments}{true}
\ifthenelse{\boolean{showcomments}}
{
  
	}

\algrenewcommand\algorithmicwhile{\textbf{upon}}
\acrodef{P2P}{peer-to-peer}

\begin{document}

\title{\sys: A Peer-Sampler with Randomness Guarantees
}

\newcommand{\sys}{\textsc{PeerSwap}\xspace}

\author{\IEEEauthorblockN{Rachid Guerraoui\IEEEauthorrefmark{1}, Anne-Marie Kermarrec\IEEEauthorrefmark{1}, Anastasiia Kucherenko\IEEEauthorrefmark{1}, Rafael Pinot\IEEEauthorrefmark{2}, Martijn de Vos\IEEEauthorrefmark{1}}
\IEEEauthorblockA{\IEEEauthorrefmark{1} \textit{Ecole Polytechnique Fédérale de Lausanne (EPFL), Lausanne, Switzerland } \\
\{rachid.guerraoui, anne-marie.kermarrec, anastasiia.kucherenko, martijn.devos\}@epfl.ch\\
\IEEEauthorrefmark{2} \textit{Sorbonne Université, Paris, France } \\
pinot@lpsm.paris}
}

\maketitle

\begin{abstract}
The ability of a peer-to-peer (P2P) system to effectively host decentralized applications often relies on the availability of a peer-sampling service, which provides each participant with a random sample of other peers.
Despite the practical effectiveness of existing peer samplers, their ability to produce random samples within a reasonable time frame remains poorly understood from a theoretical standpoint.
This paper contributes to bridging this gap by introducing \sys, a peer-sampling protocol with provable randomness guarantees.
We establish execution time bounds for \sys, demonstrating its ability to scale effectively with the network size.
We prove that \sys maintains the fixed structure of the communication graph while allowing sequential peer position swaps within this graph.
We do so by showing that \sys is a specific instance of an interchange process, a renowned model for particle movement analysis.
Leveraging this mapping, we derive execution time bounds, expressed as a function of the network size $n$.
Depending on the network structure, this time can be as low as a polylogarithmic function of $n$, highlighting the efficiency of \sys.
We implement \sys and conduct numerical evaluations using regular graphs with varying connectivity and containing up to \num{32768} ($2^{15}$) peers.
Our evaluation demonstrates that \sys quickly provides peers with uniform random samples of other peers.

\end{abstract}


\section{Introduction}
\Acf{P2P} systems are decentralized networks in which users, or \emph{peers}, communicate only directly with each other based on peer identifiers, \eg, network addresses or public keys~\cite{steinmetz2005peer}.
Many of the applications executed on a P2P system rely on so-called \emph{peer sampling services} that provide each peer with a random sample of other peers~\cite{jesi2007identifying}.
Such applications include computing aggregate functions~\cite{kempe2003gossip},  information dissemination~\cite{karp2000randomized,baldoni2007tera}, network size counting~\cite{massoulie2006peer,chatterjee2022byzantine}, consensus~\cite{bano2019sok}, clock synchronization~\cite{baldoni2009coupling}, constructing and maintaining overlay networks~\cite{voulgaris2005epidemic, jelasity2009t} and decentralized learning~\cite{devos2023epidemic}.

P2P applications critically depend on the capability of peer-sampling services to offer samples randomly and uniformly selected from the overall pool of peers.
This randomness is, for example, essential for facilitating rapid convergence in averaging tasks, improving model accuracy in decentralized learning, and faster propagation of information through the network.
Depending on the application, even a small bias in the randomness of the peer sampling service can lead to significant performance degradation at the application layer, \eg, resulting in uneven load balancing between peers, delays in information propagation, or suboptimal model convergence in decentralized learning~\cite{stutzbach2006unbiased}.

Over the past two decades, various methods to construct peer sampling services have been developed, falling into two main categories: random-walk based and gossip-based protocols.
Random-walk based peer samplers, \eg,~\cite{zhong2008convergence, bar2008rawms, gkantsidis2004random, massoulie2006peer, law2003distributed, sevilla2015node}, were used in early P2P systems like  Gnutella~\cite{ripeanu2001peer} and, more recently, in Bitcoin~\cite{nakamoto2008bitcoin}.
Essentially, the idea is to execute a random walk through the network and use the final peer encountered in this walk as the random sample.
However, random walk-based peer samplers have two main limitations: \textit{(i)} they incur a high network overhead if many samples are required~\cite{jelasity2007gossip}, and \textit{(ii)} they tend to favor nodes with high in-degrees, which compromises their ability to provide a uniform sample~\cite{massoulie2006peer}. 

Gossip-based peer samplers, \eg,~\cite{voulgaris2005cyclon,jelasity2007gossip,tolgyesi2009adaptive}, have significantly evolved in recent years, successfully overcoming the aforementioned limitations. 
The essence of the gossip-based approach involves each peer maintaining a small subset of other peers, known as a \emph{neighborhood}.
Periodically, each peer exchanges a subset of its neighborhood with another peer.
The simplicity and rapid provision of new peer samples have driven the popularity of gossip-based peer samplers.
Several works extend gossip-based peer samplers to support peers joining and leaving the network (churn)~\cite{allavena2005correctness}, or to handle the presence of malicious peers~\cite{auvolat2023basalt, antonov2023securecyclon, pigaglio2022raptee, badishi2006exposing, johansen2006fireflies, bortnikov2008brahms}.

In a gossip-based peer sampler, each peer's neighborhood strongly correlates with its initial neighborhood during the first few exchange rounds.
As the neighborhood of each peer undergoes regular updates, the distribution of potential peers in a fixed peer's neighborhood approaches uniformity, eventually providing an almost ideally random sample from the entire network.
This \emph{convergence to uniformity} is crucial, enabling P2P applications to rely on the randomness of peer samples.
Empirical evidence conveys convergence to almost uniform samples within a small, often logarithmic number of rounds in terms of the total number of peers~\cite{kermarrec2011converging, voulgaris2005cyclon,jelasity2007gossip}.

Despite their excellent practical performance, existing peer-sampling services lack a comprehensive theoretical analysis of their randomness guarantees and convergence time.
While some progress has been made in proving the  \emph{eventual uniformity} of samples~\cite{busnel2011uniformity, bortnikov2008brahms, kermarrec2011converging}, determining the time required for peer samplers to achieve uniformly random samples remains an open challenge.
Without clearly defining the convergence time, \ie, the time needed to provide a random sample to each peer, the exact performance and guarantees of peer sampling services remains unknown. 

Analyzing the randomness of peer samplers is a challenging endeavor, primarily due to the absence of explicit theoretical tools to analyze peers’ neighborhoods in dynamic networks.
One approach is to demonstrate the convergence of the entire network to a random graph, thereby asserting the randomness of each neighborhood.
However, this method lacks reasonable guarantees regarding the convergence time.
For instance, while eventual convergence to a random graph has been demonstrated in~\cite{kermarrec2011converging}, specific time guarantees are absent. 
A similar approach is also explored in the flip process~\cite{mahlmann2005peer, feder2006local}, a substantial area in graph theory.
The flip process is equivalent to a gossip-based peer sampler where peers exchange exactly one neighbor at each step.
The best-known result for the convergence time needed for the network topology to resemble a random graph is a large polynomial, around $O(n^{16})$ where $n$ is the number of peers~\cite{10.1145/1582716.1582742}.
This is an impractical result in large-scale networks where this convergence time will be disproportionally large.
Unfortunately, a specific analysis of neighborhood characteristics instead of the network graph is difficult because of complicated network dynamics.

\textbf{Our contributions.}
This paper is a first step towards bridging the gap between practical advancements and theoretical guarantees of gossip-based peer samplers.
Concretely, we make the following contributions:
\begin{enumerate}

    \item We introduce \sys(Section~\ref{section:algorithm}), a gossip-based peer-sampling service that \emph{provably and efficiently provides a random sample for each peer}.
    Similarly to other gossip-based peer samplers,    \sys has peers randomly contacting each other and exchanging neighbors.
    In contrast to existing approaches, 
    \sys \textit{(i)} has peers exchanging their entire neighborhoods instead of subsets; \textit{(ii)} maintains bidirectional connections for all peers; \textit{(iii)} uses Poisson clocks on each connection to schedule regular neighborhood exchanges (every Poisson clocks periodically rings after random time intervals to trigger a neighborhood exchange).
    
    \item We prove that the design of \sys provides an important property throughout the execution, namely that the overall \emph{connections network maintains the same underlying pattern as the initial one} despite changes in neighborhoods and connections (Section~\ref{section:interchange}). This property serves a dual purpose. Firstly, it guarantees that any network property satisfied at the start of the algorithm is preserved throughout its execution, \eg, the network connectivity, expansion value, degree distributions, and diameter.
    This can be essential for applications whose performance critically relies on the communication patterns between peers, like decentralized learning~\cite{devos2023epidemic}.
    Secondly, combined with the behavior of Poisson clocks, this property enables a comprehensive analysis of the convergence time of \sys.
    
    \item By establishing a link between the dynamics of peers in \sys and an interchange process, an important and extensively studied interacting particle system~\cite{lazarescu2015physicist, liggett1999stochastic, wood2020combinatorial}, we derive the convergence time of \sys (Section~\ref{section:mixing-time}). The upper bound we deduce depends solely on the network size ($n$) and its connectivity (measured by the spectral gap). We show that in sufficiently connected networks \emph{the convergence time is as low as polylogarithmic in $n$}.

    \item We design a variant of the \sys algorithm that temporarily locks peers involved in a neighborhood exchange. This variant can be deployed in practical settings with network delays (Section~\ref{sec:peerswap_with_delays}).

    \item We conduct numerical evaluations of \sys using regular graphs with varying connectivity and with up to \num{32768} ($2^{15}$) peers (Section~\ref{section:experiments}).
    Our evaluation conveys the convergence of \sys over time and empirically demonstrates that \sys is an efficient peer sampler.
    We also analyze the convergence and performance of our practical variant of \sys.
\end{enumerate}

\section{System Model and Background}
\label{sec:system_model}
\textbf{Preliminaries on graph theory.} 
A directed graph $G$ consists of a set of nodes $V$ and edges $E$ where each edge $(i, j)$ goes from node $i\in V$ to node $j\in V$. When $(i, j)\in E$, we say that node $j$ is adjacent to node $i$. We denote by $A_G$ the adjacency matrix of the directed graph $G$. For each node $i\in V$, \emph{the neighborhood of $i$} denoted by $N(i)$ is a set of all nodes which are adjacent to $i$. The degree of $i$, sometimes called the out-degree and denoted by $\deg(i)$, is the size of its neighborhood, or $|N(i)|$. If all nodes have the same degree $d$, the graph is called $d-$regular. If, for any pair of nodes $i, j\in V$,  we have that $i\in N(j)$ if and only if $j \in N(i)$, then we say that \emph{the graph $G$ is undirected}. We use the previous notations similarly for undirected graphs.

An important measure of graph connectivity is its spectral gap. To define it, consider $D$, the $n\times n$ matrix with all elements equal to zero, except of diagonal $(\deg(i_1), \dots, \deg(i_n))_{i_k \in V}$. \emph{The spectral gap} of the graph $G$ is: \[\lambda(G)=\lambda(W_G)= 1 - \max(|\lambda_2(W_G)|, |\lambda_n(W_G)|),\] where $1 =\lambda_1(W_G)\geq\lambda_2(W_G)\geq\dots\geq \lambda_n(W_G)\geq -1$ are eigenvalues of $W_G=D^{-1/2}A_GD^{-1/2}$. \emph{A higher spectral gap $\lambda$ indicates a better-connected graph.}

\subsection{Peer-sampling}
\label{sec:peer-sampling_systems_model}
\textbf{Peers network.}
We consider a network with $n$ fixed peers $U = \{u_1, \dots, u_n\}$, where each peer is assigned a unique identifier.
We refer both to a peer $ i $ and its identifier as $u_i$. 
Peers are connected over a routed network infrastructure, which enables communication between any pair of them on the condition that the sender knows the identifier of the receiver. 
We assume that all peers have access to a global clock and that time is \emph{continuous}, allowing peers to send messages or take actions at any moment.
Additionally, there are no failures in the system -- all peers remain online during protocol execution.
These assumptions align with related theoretical work on peer sampling~\cite{sevilla2015node, kermarrec2011converging, massoulie2006peer, gkantsidis2004random}.

\textbf{Network Topology.}
To describe the network structure at time $t$, we consider a time-dependent directed graph $G(t)=(V, E(t))$. Here, the set of nodes $V=[1, 2, \dots, n]$ represents the fixed set of peers $U$ --- any node $i\in V$ corresponds to a peer $u_i\in U$, and there is a directed edge $(i, j)\in E(t)$ if and only if the peer $u_i$ knows the identifier of the peer $u_j$ at the time $t$. We denote by $N(i, t)$ the (ordered) neighborhood of $i$ in graph $G(t)$ for any $i\in V, t\geq 0$. Although peer samplers typically operate within directed networks, we show in Section~\ref{sec:interchange} that in our setting, $G(t)$ remains undirected over time.

\textbf{Peer Sampling Service.} A peer-sampling service with parameter $\sample$ on the set of peers $U$ is a decentralized communication protocol that, at any time $t\geq 0$, provides each peer $u_i\in U$ with a random ordered tuple of distinct peers $Sample_i( \sample, t) = \left(u_{i_1}, u_{i_2},\dots, u_{i_{\sample}}\right)\subseteq U\backslash \{u_i\}$ such that $\{i_1, i_2,\dots, i_{\sample}\}\in N(i, t)$. 

\emph{The objective of such a service} is to make the distribution of $Sample_i(\sample, t)$ gradually approach a uniform distribution.  Specifically, we want the random variable $Sample_i(\sample, t)$ to converge in law (as $t \rightarrow \infty$) to $Uniform((U)_{\sample})$, where for any finite set $S$ and any $\sample\leq |S|$ we denote by $(S)_{\sample}$ the set of all (ordered) $\sample$-tuples of distinct elements from $S$:
\[(S)_{\sample}= \{\mathbf{s} \in S^{\sample}: \mathbf{s}[i]\neq \mathbf{s}[j], \ \forall i, j\in [\sample], i\neq j\}.\]
Here and further the elements of $(S)_{\sample}$ are denoted by boldface letters such as $\mathbf{x}$, with $\mathbf{x}[i]$ denoting the $i$-th coordinate of $\mathbf{x}$.

\textbf{Timing Mechanism.}
Gossip-based peer samplers require a timing mechanism that controls when a peer will exchange its neighbors with another peer. For this purpose, we use \emph{Poisson clocks} that ring periodically after intervals that follow an exponential distribution.
We formally define this below.

A random variable $Z$ has exponential distribution $\mathcal{E}(\alpha)$ with rate $\alpha$ if its cumulative distribution function is: 
\begin{equation*}
    F(z, \alpha) = \Pr[Z\leq z]=
    \begin{cases}
        1- e^{-\alpha z},&\text{if $z\geq 0$,}\\
        0,&\text{if $z<0.$ }
    \end{cases}
\end{equation*}

An important property of an exponential random variable $Z$ is \emph{memorylessness}, \ie, its future behavior remains unaffected by the time that has passed:
\begin{equation*}    
    \Pr[Z>s+t|Z>s]=\Pr[Z>t], \forall s, t\geq 0.
\end{equation*}

\begin{definition}[Poisson clock]\label{def:pois_clock} 
A Poisson clock with rate $\alpha$ is a sequence of random variables $(T_n)_{n\geq0}$ such that $T_0=0$ and for $n\geq 1$, we have $T_n = \sum^n_{k=1}Z_k$, where $(Z_k)_{k\geq1}$ is a sequence of independent exponential variables (i.e., $Z_1,\dots,Z_n\stackrel{i.i.d.}{\sim} \mathcal{E}(\alpha)$). The realizations of $(T_n)_{n\geq 0}$ are called \emph{the ring times} of the Poisson clock.
\end{definition}

\subsection{Continuous-time Markov Chains}
In this work, we propose a peer-sampling process whose future is affected solely by its current state and is independent of its history or the exact value of current time $t$. This property of the process is called Markov property, that we formally outline next.

Consider  a family of random variables $X=(X_t)_{t\geq 0}$ taking values in a finite state space $S$. $X$ is a \emph{continuous-time Markov chain (ctMC)} if it satisfies the Markov Property:
\begin{align*}
    \Pr[X_{t_n}=i_n\mid X_{t_1}=i_1, \dots, X_{t_{n-1}}=i_{n-1}]=\\\Pr[X_{t_n}=i_n\mid X_{t_{n-1}}=i_{n-1}].
\end{align*}
for all $i_1, \dots,i_n\in S$ and any sequence $0\leq t_1\leq t_2\leq \dots\leq t_n$ of times. 
The process is \emph{time-homogeneous} if the conditional probability does not depend on the current time, but only on the time interval between observations, \textit{i.e.}, 
\begin{align}\label{eq:transition-pr}
    \Pr[X_{t+h}=j\mid X_{h}=i]=\Pr[X_{t}=j\mid X_{0}=i],\\ \text{ for any }i, j\in S,  h\geq 0.\nonumber
\end{align}

For time-homogeneous ctMC we denote the probability in \eqref{eq:transition-pr} as \emph{a transition probability} $P_{i, j}(t)$ and define \emph{the matrix of transition probabilities} at time $t$ as $P(t)=(P_{i, j}(t))_{i, j\in S}$. A ctMC is called \emph{irreducible} if for any states $i, j\in S$ there exists $t>0$ such that $P_{i, j}(t)>0$. In this paper, if not stated otherwise, we consider irreducible time-homogeneous ctMCs. 

\vspace{10pt}

The transition probability can be used to completely characterize  the ctMC, but it gives a lot of information, which is not always needed. An alternative (and more succinct) way of characterizing the dynamic of a ctMC is the \emph{infinitesimal generator}~\cite{levin2017markov}[Chapter 20.1]. 
\begin{definition}[Infinitesimal generator]\label{def:infinitesimal}
    Let  $X=(X_t)_{t\geq 0}$ be a time-homogeneous ctMC with finite set of states $S$ and transition probabilities $P(t)$. The infinitesimal generator of $X$ is the matrix of size $|S|\times |S|$ defined as:
    \[Q=\lim\limits_{h\rightarrow 0+}\frac{P(h)-I}{h},\]
    where $I$ is the identity matrix of size $|S|\times |S|$. We denote entries of $Q$ as $q(i,j)$ for $\forall i, j\in S$. By construction of $Q$ we always have $\sum\limits_{j \in S}q(i, j)=0, \forall i\in S$.
\end{definition}

\begin{definition}\label{def:TV}
 The \emph{total variation distance} $d_{TV}$ between two distributions\footnote{We use a slight abuse of notation by interchangeably referring to a discrete distribution and the vector representing its probabilities.} $\mu, \nu$ on the same finite set of states $S$ is given by:
\begin{equation}
    d_{TV}(\mu, \nu) = 1/2\sum\limits_{s\in S}\mid \mu_s-\nu_s \mid,
\end{equation}
where $\mu_s$ (resp. $\nu_s$) denotes the probability of observing $s$ according to $\mu$ (resp. $\nu$).
\end{definition}

\begin{theorem}[Theorem 20.1 from~\cite{levin2017markov}]\label{thm:stationary}
    Consider an irreducible ctMC $X$ on a finite set of states $S$  with matrix of transition probabilities $P(t)$. Then there exists a unique distribution $\pi=(\pi_j)_{j\in S}$ such that $\pi P(t) = \pi$ for all $t\geq0$, and 
    \[\max_{ i\in S} d_{TV}(P_{i, \cdot}(t), \pi )\xrightarrow{t\rightarrow \infty}0,\]
    where  $P_{i, \cdot}(t)=(P_{i, j}(t))_{j\in S}$.
\end{theorem}

 We call such $\pi$ \emph{a stationary distribution} of $X$ and define the $\varepsilon$\emph{-mixing time of a ctMC} $X$ as 
    \[T_X(\varepsilon) = \inf \left \{ t \geq 0 : \max_{ i\in S} d_{TV}(P_{i, \cdot}(t), \pi ) \leq \varepsilon \right \}. \]
The mixing time of ctMCs is a key component in our further derivations of the convergence time of our peer sampler.

\subsection{Auxiliary process of peers' movement on a fixed graph}
\label{sec:interchange}
Peer-sampling services keep the set of peers fixed and modify the connections between them.
In this way, graphs $G(t)$ represent the consecutive states of a given ``dynamic'' graph with fixed nodes and reconnecting edges.
Usually, processes on dynamic graphs are complicated to analyze.
Thus, in this paper, we map peer-sampling processes on dynamic graphs and an auxiliary so-called interchange process where the graph is \emph{fixed} while peers change positions on the graph's nodes.

\textbf{Interchange process.}
An \emph{interchange process} $IP(k, G)$ describes dynamics of $k$ peers on $G=(V, E)$.  The set of possible states for this chain is $(V)_k$, \ie all possible positions of $k$ peers on the nodes of $G$, where no two peers can be positioned on the same node. The transition between states happens by randomly choosing edges from $E$, and if an edge $e\in E$ is chosen, the peers (if any) positioned at the endpoints of $e$ are switched. If there is just one peer, it still switches its position to the other endpoint of the edge. We will implicitly assume that $G$ is connected, in which case interchane process is irreducible and thus admits a stationary distribution and mixing time~\cite{levin2017markov}[Example 1.12.].

Now, we formally describe an interchange process as per \cite{oliveira2013mixing}. We define a transposition function for any edge $e=(i, j)\in E$ and node $v\in V$ as
\begin{equation}\label{eq:transition_f-1}
    f_e(v) = \begin{cases}
j,&\text{if $v=i$,}\\
i,&\text{if $v=j$,}\\
v,&\text{otherwise. }\end{cases}
\end{equation}
We also write 
\begin{equation}\label{eq:transition_f-2}
    f_e(\mathbf{x})=\left ( f_e(\mathbf{x}[i]) \right)_{i\in[1, k]}\text{ for } \mathbf{x}\in (V)_k.
\end{equation}

\emph{An interchange process with parameter $k$} on $G=(V, E)$, or $IP(k, G)$ for short, is a ctMC with state space $(V)_k$ such that for any distinct $\mathbf{x}, \mathbf{y} \in (V)_k$:
\begin{equation}\label{eq:interchange_transition}
    q(\mathbf{x}, \mathbf{y}) = 
    \begin{cases}
        1,&\text{if $\exists e\in E$ s.t. $f_e(\mathbf{x})=\mathbf{y}$,}\\
        0,&\text{otherwise. }
    \end{cases}
\end{equation}

\begin{figure*}[ht]
\vspace{-0.3cm}
\includegraphics[width=\textwidth]{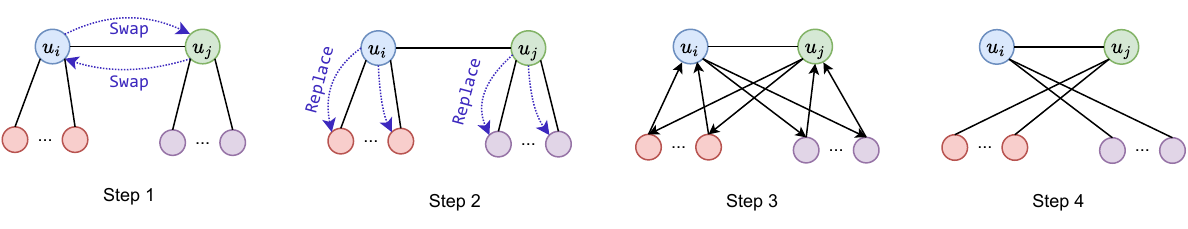}
\caption{A single swap during the \sys execution between a pair of adjacent peers $u_i$ and $u_j$. Here, blue doted lines indicate sent messages and black solid lines represent connections between peers. We indicate a directed edge with an arrow.}
\label{fig:swap_illustration}
\vspace{-0.3cm}
\end{figure*}

\begin{algorithm}[t]
\caption{The \sys algorithm from the perspective of peer $u_i$.}\label{alg:main}
\begin{algorithmic}[1]
\Require Undirected $G(0)=(V, E(0))$ describing users $U$; parameter $\sample$; Poisson clocks with rate $\alpha$ $C_e=C_i(j)=C_j(i)$ for $\forall e=(i, j)\in E(0)$.
\Ensure $Sample_i(\sample, t)\sim Uniform((U\backslash\{u_i\})_{\sample})$. 
\Statex
\While{some $C_e$ rings between $i$ and $j$ at time $t$}\label{alg:line:3}
\State Set \textsc{Swap$_i$} $\leftarrow ((u_{\kappa}, C_i(\kappa))_{\kappa\in N(i, t)})$
\State Send \textsc{Swap$_i$} to $u_j$ \Comment{Step 1}
\State Set \textsc{Replace$_j$} $\leftarrow u_j$
\State Send \textsc{Replace$_j$} request to all $(u_{\kappa})_{\kappa\in N(i, t)}$ \Comment{Step 2}
\EndWhile \label{alg:line:4}

\Statex
\While{receiving $\textsc{Swap}_j$ from $u_j$ at time $t$} \Comment{Step 3}
\State Erase $((u_{\kappa}, C_i(\kappa))_{\kappa\in N(i, t)})$
\State Store $((u_l, C_j(l))_{l\in N(j, t)})$ from $\textsc{Swap}_j$
\Statex\Comment{\emph{$G(t)$ changed, $N(i, t)\leftarrow N(j, t)$}}
\EndWhile

\Statex
\While{receiving $\textsc{Replace}_l$ from some $u_{\kappa}$} \Comment{Step 4}
    \State Store $u_l$ and set $C_i(l)\leftarrow C_i(\kappa)$
    \State Erase $(u_{\kappa}, C_i(\kappa))$
    \Statex\Comment{\emph{$G(t)$ changed: in $E(t)$ $(i, \kappa)$ is replaced with $(i, l)$}}
\EndWhile \label{alg:line:13}

\Statex

\Function{\textit{Sample}$_i$}{$\sample,t$}\label{alg:line:14}
\State At time $t$ select $\sample$ random peers from $(u_{\kappa})_{\kappa\in N(i, t)}$ 
\State Return a randomly ordered tuple of these peers
\EndFunction \label{alg:line:15}

\end{algorithmic}
\end{algorithm}

\section{Poisson-based random swap (\sys)}
\label{section:algorithm}

We now introduce the \sys protocol.
An execution of \sys starts from a set of peers $U$ connected in a $d$-regular\footnote{For presentation clarity, we focus on regular topologies. However, \sys and its analysis are also applicable to irregular graphs.} undirected graph $G(0)=(V, E(0))$, where any node $i\in V$ represents a user $u_i\in U$. 
\sys assigns an independent Poisson clock $C_e$ with rate $\alpha$ to each edge $e\in E(0)$.
By default, $\alpha =1$, and system designers can adjust $ \alpha $ to speed up or slow down swap rates based on application requirements (more details in Section~\ref{subsec:convergence_peerswap_with_network_delays}.
Each clock $C_e$ is shared between the peers at the endpoints of the edge $e \in E(0)$. 
If no Poisson clock rings, the graph $G(t)$ remains unchanged as time passes.
Whenever a Poisson clock on some edge $e=(i, j)$ rings at time $t$, peers $ u_i $ and $ u_j $ swap their position following the four steps visualized in Figure~\ref{fig:swap_illustration} and as described in Algorithm~\ref{alg:main}: 
\begin{itemize}
    \item \textbf{Step 1:} Peer $u_i$ (resp. $u_j$) sends a \textsc{Swap} message initiating the swap to $u_j$ (resp. $u_i$). This \textsc{Swap} message contains identifiers of all users $u_{\kappa}$ such that $\kappa\in N(i, t)$ (resp. $N(j, t)$), as well as all the Poisson clocks $u_i$ (resp. $u_j$) shares with its neighbors.
    \item \textbf{Step 2:} Peer $u_i$ sends a $\textsc{Replace}_j$ message to its neighbors $(u_{\kappa})_{\kappa\in N(i, t)}$, informing them about the swap and requesting them to replace its identifier, \ie, $u_i$, with $u_j$.
    Similarly, peer $u_j$ also sends a $\textsc{Replace}_i$ message to its neighbors with the identifier of peer $u_i$.
    \item \textbf{Step 3:} Upon reception of a \textsc{Swap} message, peer $u_i$ (resp. $u_j$) overwrites its neighborhood and associated Poisson clocks it previously stored with the ones contained in the received message.
    This step changes the structure of the current graph $G(t)$, as depicted in Step 3 in Figure~\ref{fig:swap_illustration}.
    During this step, the graph may appear temporarily undirected.
    \item \textbf{Step 4:} Peers that received a $\textsc{Replace}_j$ message from peer $ u_i $ replace the identifier of $u_i$ with that of $u_j$.
    Additionally, consider peers that were in the former neighborhood of $u_i$ ($(u_{\kappa})_{\kappa\in N(i, t)}$) that now became a neighborhood of $u_j$. These peers also replace identifier $u_i$  with $u_j$, but without modifying the associated Poisson clock.
\end{itemize}

\sys is driven by the Poisson clocks placed on the graph edges. 
We explain how to implement them in a distributed manner in Section~\ref{sec:peerswap_with_delays}.
The rate $\alpha$ indicates that, on average, two peers connected by an edge swap their positions after every time period $\alpha$.
Thus, peers continuously exchange their neighborhoods and get their neighborhoods refreshed.
At any point in time $t \geq 0$, a peer $ u_i $ can invoke the $Sample_i(\sample, t)$ function, yielding a random sample of $\sample$ peers from the neighborhood of $u_i$ at time $t$ (Lines \ref{alg:line:14}-\ref{alg:line:15} in Algorithm~\ref{alg:main}).

\begin{remark}
    When we refer to $G(t)$ we always mean its final configuration. This means that if a Poisson clock rang at time $t$, $G(t)$ will always refer to the graph after the modification. 
\end{remark}

For the upcoming theoretical analysis, we assume that the system has no delays and all actions during the four steps described above are executed instantly.
Thus, by design, for any $t \geq 0$, the graph $G(t)$ is always undirected.
In Section~\ref{sec:peerswap_with_delays} we describe how \sys can be adapted to function in more realistic settings with network delays.

\section{The Equivalence of \sys and the Interchange Process} \label{section:interchange}
By considering the swap in Figure~\ref{fig:swap_illustration} and changing the graph representation at the end of Step 4 in Figure~\ref{fig:alternative}, we observe that running the swap procedure between two peers, \ie, modifying the structure of the graph, is equivalent to having a fixed graph structure on which we exchange the positions of the two considered peers.
In this section, we prove that this equivalence holds in general by relating \sys (which produces a dynamic sequence of graphs $G(t)$) to an interchange process on $G(0)$. This equivalence is the main building block of our analysis and allows us to obtain the convergence result for $Sample_i(b,t)$ in Section~\ref{section:mixing-time}. Complete proofs related to this section can be found in Appendix~\ref{appendix:proofs1}.

\subsection{Isomorphism of graphs built by \sys}

\vspace{5 pt}
\fbox{\begin{minipage}{22em}
\emph{During \sys process, the network can be described by a sequence of graphs $(G(t))_{t\geq 0}$ where the set of nodes $V$ represents peers (remains fixed), and edges $E(t)$ are modified.} \end{minipage}}

\vspace{5 pt}

Below, we show that for any $t\geq0$ their exists an isomorphism between $G(t)$ and $G(0)$.

\begin{definition}[isomorphism]
    Two graphs $G$ and $G'$ are isomorphic according to a bijection $\sigma: V(G)\rightarrow V(G')$ if for any two vertices $i$ and $j$ that are adjacent in $G$, $\sigma(i)$ and $\sigma(j)$ are adjacent in $G'$ (and vice versa).
\end{definition}

\begin{restatable}{lemma}{FirstBijection}
\label{lemma:transition-swap}
     For all $t\geq 0$, $G(0)$ and $G(t)$ are isomorphic according to a bijection $\sigma_{t}:V(G(0))\rightarrow V(G(t))$. Furthermore, if some nodes $i$ and $j$ are adjacent in $G(t)$, then the Poisson clock they share is $C_{e}$ with $e=(\sigma_{t}^{-1}(i), \sigma_{t}^{-1}(j))\in E(0)$.
\end{restatable}

Now that we know that graphs are isomorphic and that all Poisson clocks of the dynamic graphs $(G_t)_{t>0}$ can be mapped to Poisson clocks on the initial graph $G(0)$, we show that there exists an alternative bijection from $G(t)$ into $G(0)$ defined only by Poisson clocks on $G(0)$. This bijection is easier to operate with and is going to be used for the rest of the paper.

\begin{restatable}{lemma}{MainBijection}
  Consider \sys process starting from graph $G(0)=(V, E(0))$, resulting in sequence of graphs $(G(t))_{t\geq 0}=(V, E(t))_{t\geq0}$. Let us fix some value $t>0$ and let ring times of Poisson clocks up to time $t$ be $t_1, t_2, \dots, t_N$ with $0=t_0<t_1<\dots<t_N\leq t$. Denote $(i_m, j_m)=e_m\in E(t_{m-1})$ the edge on which the clock rang at time $t_m$, and let $(\ih_m, \jh_m)=\eh_m = (\sigma_{t}^{-1}(i_m), \sigma_{t}^{-1}(j_m))\in E(0)$ be the corresponding edge to $e_m\in E(0)$. We define
  \[ \gamma_t = f_{\eh_N} \circ \dots \circ f_{\eh_1},\]
    where for every $e$, $f_{e}$ is as defined in \eqref{eq:transition_f-1}.
    
    Then $\gamma_t=\sigma^{-1}_t$, where $\sigma_t$ is as defined in Lemma\ref{lemma:transition-swap}.
\end{restatable}

\begin{figure*}
\centering
\vspace{-0.3cm}
\includegraphics[width=.9\textwidth]{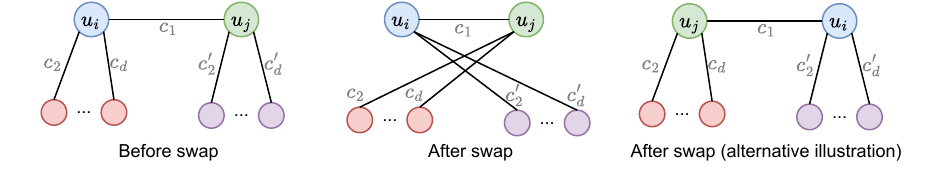}
\caption{The network structure before and after a swap in \sys between a pair of adjacent peers $u_i$ and $u_j$ (as also shown in Figure~\ref{fig:swap_illustration}). $c_1, c_2,\dots, c_d, c'_2,\dots, c'_d$ correspond to Poisson clocks assigned to the respective edges. The alternative illustration after the swap shows that the general network pattern did not change, but peers $u_i, u_j$ switched their positions.  }
\label{fig:alternative}
\vspace{-0.3cm}
\end{figure*}

\subsection{Alternative \sys process and its connection to interchange}

Due to the isomorphism of $(G(t))_{t\geq0}$ and $G(0)$, studying the dynamic of  $(G(t))_{t\geq0}$ is equivalent to studying an auxiliary process where peers move on a fixed graph $G(0)$. Formally this process is described below.

\begin{definition}\label{def:main_process}
We study connections between any $k$ distinct peers $(u_{i_1}, u_{i_2}, \dots u_{i_k})\subset U$, such that $n\geq k>0$. To do so, we consider the process $Y^{(k)}=(Y^{(k)}_t)_{t\geq 0}$.  It has a set of states $(V)_k$, initial state $Y^{(k)}_0=(i_1, i_2, \dots, i_k)$ and for any $t \geq 0$ we set $Y^{(k)}_t = \gamma_t(Y^{(k)}_0)$, where $\gamma_t(\mathbf{x})=\left ( \gamma_t(\mathbf{x}[i]) \right)_{i\in[1, k]}$ for any $\mathbf{x}\in (V)_k$. 
\end{definition}

\vspace{5 pt}

\fbox{\begin{minipage}{23em}
\emph{Studying processes $Y^{(k)}$ is an alternative way to describe dynamic of \sys. $Y^{(k)}$ operates on a fixed graph $G(0)=(V,E(0))$, and peers $U$ or their subset are switching positions on top of $V$.}
\end{minipage}}

\vspace{5 pt}

We show that for any $k$ process $Y^{(k)}$ is a time-homogeneous ctMC and prove that it is an interchange process.

\begin{restatable}{theorem}{SwapIsInterchange}\label{thm:swap-is-interchange}
     Consider \sys process on peers $U$ starting from graph $G(0)=(V, E(0))$, and process $\swap^{(k)}$ characterizing some $k$ peers in \sys, $1\leq k\leq |V|$ as per Definition~\ref{def:main_process}. Then, $\swap^{(k)}$ is an interchange process $IP(k, G(0))$ (as defined in Section \ref{sec:interchange}.)
     
\end{restatable}

\begin{proof}[Proof sketch]
We demonstrate that since isomorphisms $\gamma_t$ rely on Poisson clocks that are memoryless,  $\swap^{(k)}$ is a time-homogeneous ctMC. Then, to show that two ctMCs are equivalent, it is sufficient to prove that their infinitesimal generators are equal (as in Definition~\ref{def:infinitesimal}). We calculate the probabilities of zero, one, and more than one Poisson clock ringing in any period of length $h$. These probabilities allow to derive $ \lim\limits_{h\rightarrow 0+}\frac{P_{\mathbf{x}, \mathbf{y}}(h)}{h} = q(\mathbf{x}, \mathbf{y})$ for any two states $\mathbf{x}, \mathbf{y}\in (V)_k$ and to show the equality with \eqref{eq:interchange_transition}. 
\end{proof}

\section{Convergence time of \sys}
\label{section:mixing-time}

Since \sys is tightly connected to an interchange process, we can study the neighborhood of each peer in an execution of \sys on $G(0)$ by considering $IP(n, G(0))$. By utilizing this, we bound the difference between the probability of any particular peer sample that \sys can provide and the uniform probability, \ie, $\frac{1}{\binom{n-1}{d}}$. Complete proofs related to this section can be found in Appendix~\ref{appendix:proofs2}.

\begin{restatable}{theorem}{FirstConvergenceTheorem}\label{thm:main-1}
     Consider a set of peers $U=\{u_1, \dots, u_n\}$ connected in a $d$-regular connected graph $G(0) =(V, E(0))$, such that $V =\{1, 2, \dots, n\}$ and each node $i$ corresponds to peer $u_i$. Also, consider \sys protocol  starting from $G(0)$ with Poisson clocks of rate $\alpha$. Then for any $\varepsilon>0$, any peer $u_i\in U$, and any $\mathbf{u}\in (U\backslash \{u_i\})_d$ the following holds true \[\left|\Pr[Sample_i(d, T)=\mathbf{u}]- \frac{(n-d-1)!}{(n-1)!}\right|\leq 4\varepsilon,\]
    where $T:= \alpha\cdot T_{Y^{(d+1)}}(\varepsilon)$ and $Y^{(d+1)}$ is an $IP(d+1, G(0))$ i.e., $\varepsilon$-mixing time of an interchange process $IP(d+1, G(0))$.
\end{restatable}

\begin{proof}[Proof Sketch] We start by showing that for bounding the left side of theorem inequality it is sufficient to bound the difference between the probabilities that $\mathbf{u}$ and any other sample $\mathbf{u'}\in (U\backslash\{u_i\})_d$ are neighborhoods of $u_i$. Then, we use the fact from Theorem\ref{thm:main-1} that connections between $u_i$ and $\mathbf{u}$ (resp. $\mathbf{u'}$) in \sys follow the same dynamic as those of peers in $Y$ (resp. $Y'$), interchange processes $IP(d+1, G(0))$. Processes  $Y$ and $Y'$ have the same stationary distributions and mixing times $T=T_Y(\varepsilon)=T_{Y'}(\varepsilon)$. Hence, we consider $\rho$ -- the stationary probability of the event that peers are located in such a way that the first one has all others in the neighborhood. Then after time $T$ both $\Pr[Sample_i(d, T)=\mathbf{u}]$ and $\Pr[Sample_i(d, T)=\mathbf{u'}]$ differ from $\rho$ no more than by $\varepsilon$. This implies that the difference between these probabilities is no more than $2\varepsilon$. The amount of different samples is $\frac{(n-1)!}{(n-d-1)!}$ that concludes bounding of the theorem expression.
\end{proof}

The above result ensures that distribution of $Sample_i(d, t)$ converges to a uniform distribution on $(U\backslash\{u_i\})_d$ for any $u_i\in U$.
We now assess the convergence time of \sys.

\begin{restatable}{theorem}{FinalConvergenceTheorem} \label{th:main-2}
    Consider a set of peers $U=\{u_1, \dots, u_n\}$ connected in a $d$-regular connected graph $G(0) =(V, E(0))$ and \sys protocol executing starting from $G(0)$ with Poisson clocks of rate $\alpha$. For any peer $u_i\in U$, we consider the random variable $Sample_i(d, t)$ after time $t$ and denote by $L(Sample_i(d, t))$ its probability distribution. Then for any $\delta\in(0, 1)$ there exists $T\in O\left(\frac{\alpha \log(n)\log(n^{d+1}/\delta)}{d\lambda(G(0))}\right)$ such that:
    \[d_{TV}\left(L(Sample_i(d, T)), Uniform((U\backslash\{u_i\})_d)\right)\leq \delta.\]% \text{ for any } t \geq T.\]
\end{restatable}

\begin{proof}[Proof Sketch]
To calculate the total variation distance described in the theorem statement, we write its value by Definition~\ref{def:TV} and use the result of Theorem~\ref{thm:main-1}. Then, we estimate the value of $T$ by bounding the mixing time of an interchange process $IP(d+1, G(0))$. This bound is possible due to the connection of interchange processes and random walks. Random walks are ctMCs with mixing time that is directly related to the graph connectivity. 
\end{proof}

\textbf{Implications of Theorem~\ref{th:main-2}.}
Hence, when $\delta$ and $\alpha$ are constants of $n$, the distribution of $Sample_i(d, t)$, $\forall i\in V$ is $\delta$-close to the uniform one after time $T= O\left(\frac{\log^2(n)}{\lambda(G(0))}\right)$, which is a simpler expression than the bound derived above. Moreover, if the graph is sufficiently-connected, \ie, the spectral gap of $G(0)$ is $\Omega(\frac{1}{\log^c(n)})$, where $c$ is a constant (applicable to many popular graphs such as hypercubes, Erdos-Renyi or expander graphs), then the convergence time of \sys on network $G(0)$ is polylogarithmic.
Lastly, we note that our upper bound on convergence time with sample size $d$ also applies to any other sample size $b$ where $1\leq b\leq d$.

\section{A Practical Implementation of \sys}
\label{sec:peerswap_with_delays}
So far, we have presented the \sys protocol and theoretically proven its convergence.
In this section, we first describe how one can implement a shared Poisson clock in a decentralized manner and then present a variant of Algorithm~\ref{alg:main} that can function in a setting with network delays.

\subsection{Decentralized Poisson Clocks}
Algorithm \ref{alg:main} assigns a Poisson clock to each edge to schedule swaps.
If a central coordinator supervises the network topology, it can keep track of these clocks and inform the peers at the endpoints of an edge when a swap occurs (\ie, a clock rings).
In networks without centralized control, however, realizing a Poisson clock on each edge requires coordination and synchronization between the peers at both ends of the edge.
To address this, we ensure that both peers at the endpoints of a particular edge $ e $ construct a common Poisson clock for $ e $ by initializing and operating a pseudorandom number generator with a common seed.
We refer the reader to Appendix~\ref{app:decentralized_poisson_clocks}, where we outline this subprotocol.

\begin{figure}[b]
\centering
\includegraphics[width=.85\linewidth]{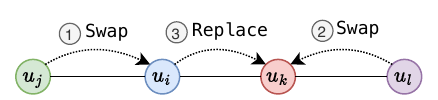}
\caption{A situation where the execution of Algorithm~\ref{alg:main} in settings with network delays can violate protocol correctness and alter the fixed structure of the graph.}
\label{fig:peerswap_fail}
\end{figure}

\subsection{The Challenge of \sys with Network Delays}
Deploying Algorithm~\ref{alg:main} in P2P networks with network delays may fail, as network delays affect the ordering of messages received by peers and thus break the fixed structure of the underlying topology.
This is because the execution of a single swap $ A $ consists of two sequential message deliveries (\textsc{swap} and \textsc{replace}), each message taking some time to be delivered if network delays are present.
Meanwhile, the execution of another swap may interfere with $ A $.
We elaborate this with an example, involving four peers $ u_i $, $ u_j $, $ u_k $ and $ u_l $, as shown in Figure~\ref{fig:peerswap_fail}.
When the Poisson clock associated with the edge $ (u_i, u_j) $ rings, $ u_j $ will send a \textsc{Swap} message to $ u_i $ (for presentation clarity, we do not show the \textsc{Swap} message by $ u_i $).
Now assume that nearly at the same time, the Poisson clock associated with the edge $ (u_k, u_l) $ rings.
Peer $ u_l $ thus sends a \textsc{Swap} message to $ u_k $.
Assume that $ u_k $ receives this swap message before the \textsc{Replace} message from $ u_i $, the latter being associated with the swap between $ u_i $ and $ u_j $.
Now, $ u_k $ will replace its neighbors with the neighbors of $ u_l $.
When $ u_k $ later receives the \textsc{Replace} message from $ u_i $, $ u_k $ might now be unable to correctly replace $ u_i $ in its neighbor list, resulting in an execution error.
More specifically, the \textsc{Replace} message sent by $ u_i $ to $ u_k $ should have been sent to $ u_l $ instead.

\subsection{Lock-based Peer Swaps}
Generally speaking, it is unsafe to execute the \sys protocol for any two swaps involving at least one common peer.
We address this concern by having peers \emph{lock other peers prior to executing a swap}.
In this extended protocol, a peer $ u_i $ that is locked for some swap $ A $ will only process messages related to $ A $.
When a swap between $ u_i $ and $ u_j $ starts, both $ u_i $ and $ u_j $ attempt to lock their neighbor peers.
If $ u_i $, $ u_j $, or any of their neighbors is already locked for another swap, the swap fails.
Despite the fact that some swaps might fail, we found this simple protocol to be effective and avoiding the need for more resource-intensive forms of coordination such as consensus or other agreement protocols.
We next describe the involved steps in this extended protocol that makes \sys suitable for deployment in settings with network delays.
We explain the execution of a particular swap $ A $ between peers $ u_i $ and $ u_j $, and outline the scenarios where a swap can fail.

\textbf{Preparing for a swap.}
Each peer keeps track of its lock status, \ie, if it has been locked for some swap.
Whenever a Poisson clock on some edge $ e = (i, j) $ rings, $ u_i $, respectively $ u_j $, checks if it is already locked.
If $ u_i $ is not locked for another swap $ A' \neq A $, it will lock itself for $ A $.
$ u_i $ then request its neighbors to lock themselves as well for $ A $ by sending a \textsc{LockRequest} message to all its neighbors except $ u_j $.
This message includes the specifications of swap $ A $.

When peer $ u_k $ receives a \textsc{LockRequest} message from $ u_i $, $ u_k $ will check if it is already locked for another swap $ A' \neq A $.
If not, it will lock itself for $ A $ and reply with a \textsc{LockResponse} message.
This message contains a binary flag \emph{success} that indicates whether $ u_k $ has successfully locked itself in response to the received \textsc{LockRequest} message from $ u_i $.
If $ u_k $ is not locked, \emph{success} will be set to true, whereas if $ u_k $ is already locked for another swap, it will be set to false.

Upon reception of a \textsc{LockResponse} message by $ u_i $ from $ u_k $, $ u_i $ verifies if it is still participating in swap $ A $, \ie, is locked for this swap, and ignores the message if it is unlocked.
This check is necessary because swap $ A $ could have already failed (explained below) even though the \textsc{LockResponse} message from $ u_k $ (or the \textsc{LockRequest} message to $ u_k $) was still in transit.
If all neighbors of $ u_i $ replied with a \textsc{LockResponse} message containing a positive value for the \textsc{success} field, $ u_i $ sends a \textsc{Swap} message to $ u_j $.
However, if even a single \textsc{LockResponse} message with a negative value for the \textsc{success} field is received, the swap cannot proceed.
In this situation, $ u_i $ sends an \textsc{Unlock} message to the neighbors it previously sent a \textsc{LockRequest} to, to ensure that these previously-locked neighbors can participate in other swaps.
Furthermore, it sends a \textsc{SwapFail} message to $ u_j $, informing about the failed swap.
When $ u_j $ receives a \textsc{SwapFail} message from $ u_i $, it also sends an \textsc{Unlock} message to its neighbors, except sending it to $ u_i $.
$ u_j $ will then unlock itself.

\textbf{Executing a swap.}
If $ u_i $ has received positive \textsc{LockResponse} messages from all its neighbors, as well as received a \textsc{Swap} message from $ u_j $, swap $ A $ is now safe to execute as all peers involved in $ A $ are aware of the swap and locked for it.
At that point, $ u_i $ sends \textsc{Replace} messages to its neighbors, similarly as described in Algorithm~\ref{alg:main}.
Then, $ u_i $ erases its neighborhood and stores identifiers of the neighbors of $u_j$ contained in the received \textsc{Swap} message.
Finally, $ u_i $ unlocks itself as it has completed all required actions for swap $ A $.
When $ u_k $ receives a \textsc{Replace} message, it will replace $ u_i $ with $ u_j $ in its neighborhood and unlock itself.
This completes the execution of swap $ A $.

\section{Numerical Evaluation}
\label{section:experiments}
We present numerical evaluations that analyze the convergence speed of \sys under different parameters and network topologies.
We have implemented a simulation of \sys in the Python programming language, using the \textsc{networkx} library to generate network topologies.
This simulator includes an implementation of our lock-based swap protocol described in Section~\ref{sec:peerswap_with_delays}.
All our source code, experiment scripts, and documentation is publicly available.\footnote{Source code available at \url{https://github.com/sacs-epfl/peerswap}.}

\subsection{Setup and Convergence Metric}
The main objective of our evaluation is to empirically show that \sys quickly provides peers with a neighborhood of peers that is indistinguishable from taking a uniform random sample from all peers in the topology.
To do so, we run \sys with different parameters and pre-defined experiment durations $ T$.
If not stated otherwise, the default rate of Poisson clocks is $1$.
At the end of each experiment, we observe the neighborhood of one or more peers.
We run each experiment multiple times (we specify how many for each experiment in this section), which yields the frequencies of final neighborhoods encountered for each peer. 
In perfect circumstances, each peer encounters each neighborhood with equal frequencies; however, these frequencies should show some variance in practice.
To understand the convergence of \sys, we disable network latencies for all upcoming experiments, except for those in Section~\ref{subsec:convergence_peerswap_with_network_delays}.

To determine convergence in \sys, we run a Kolmogorov–Smirnov (KS) test~\cite{berger2014kolmogorov} between the observed frequencies of neighborhoods and a synthetic sequence of neighborhood frequencies that we construct by uniformly sampling all possible neighborhoods.
The KS-test is a standard test that is frequently used to compare the similarity of distribution functions.
In the context of \sys, it indicates how far the neighborhood frequencies of a particular peer differ from a uniform random distribution.
This test yields two values: a distance and a $p$-value. 
In our experiments, the KS-test distance quantifies the maximal deviation between the cumulative distribution functions of the observed and expected neighborhood frequencies, indicating the extent of similarity.
The $p$-value essentially assesses the statistical significance of this deviation.
More precisely, it represents the probability of wrongly rejecting the null hypothesis "the two distributions are identical."
A lower $p$-value thus suggests less likelihood that the observed distribution matches a uniform distribution.

The following experiments are presented in three parts.
In Section~\ref{subsec:exp_varying_connectivity}, we provide an extensive analysis of neighborhoods' frequencies (of $Sample_i(\sample, T)$, when $\sample=d$) when running \sys in small-scale networks, \ie, $n\leq 100$ peers.
In Section~\ref{subsec:exp_convergence_large_topologies}, we present a peer-frequency analysis (of $Sample_i(\sample, T)$, when $\sample=1$) on large networks, \ie, $n\geq 1000$ peers.
Finally, in Section~\ref{subsec:convergence_peerswap_with_network_delays}, we explore the convergence and performance of \sys with different forms of network delays.

\begin{figure}[t]
\centering
\includegraphics[width=\linewidth]{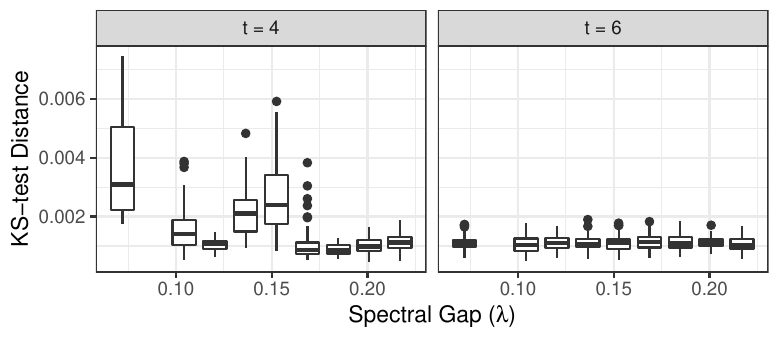}
\caption{The distribution of KS-test distances for neighborhoods distributions, when running \sys for topologies with different spectral gaps ($ \lambda $) and experiment durations ($ T $).} 
\vspace{-0.3cm}
\label{fig:ks_distances}
\end{figure}

\begin{figure*}[ht!]
	\centering
	\begin{subfigure}{.36\linewidth}
		\centering
		\includegraphics[width=\columnwidth]{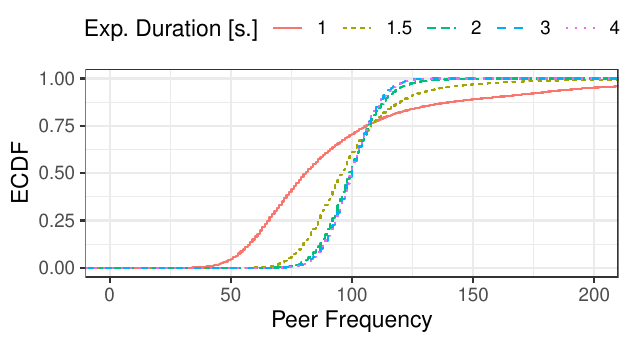}
        \captionsetup{width=.9\linewidth}
        \caption{Peer frequencies after various durations ($ n = 4096, d = 5 $).}
		\label{fig:nodes_frequencies_different_t_ecdf}
	\end{subfigure}%
    \begin{subfigure}{.38\linewidth}
		\centering
		\includegraphics[width=\columnwidth]{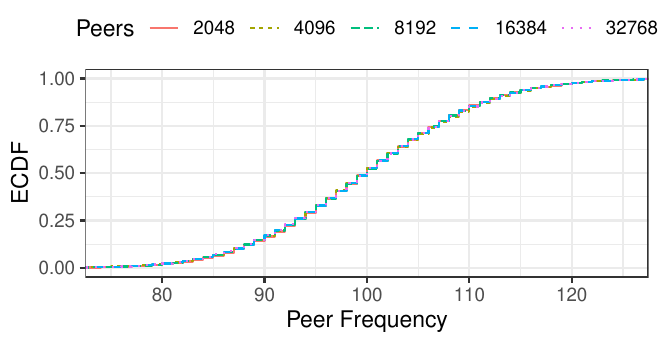}
        \captionsetup{width=.9\linewidth}
        \caption{Peer frequencies for different network sizes ($ d = 5, T = 5 $).}
		\label{fig:nodes_frequencies_different_n_ecdf}
	\end{subfigure}%
    \begin{subfigure}{.26\linewidth}
		\centering
        \footnotesize
		\begin{tabular}{c|c|c}
        \toprule
        \textbf{Peers ($n$)} & \textbf{\makecell{KS-test \\ distance}} & \textbf{\makecell{KS-test \\ $p$-value}} \\ \hline
        \num{2048} & 0.020 & 0.79 \\
        \num{4096} & 0.011 & 0.96 \\
        \num{8192} & 0.012 & 0.65 \\
        \num{16384} & 0.006 & 0.88 \\
        \num{32768} & 0.003 & 0.97 \\
        \bottomrule
    \end{tabular}
        \captionsetup{width=.9\linewidth}
        \caption{KS-test statistics associated with plot (b), for different values of $ n $.}
		\label{tab:ks_test_large_topologies}
	\end{subfigure}%
	\caption{ECDFs of peer frequencies' distributions across diverse experiment durations and larger topologies (up to \num{32678} nodes). For $T\geq 3$, ECDFs are visually indistinguishable from those obtained through uniform sampling. }
    \vspace{-0.3cm}
	\label{fig:emperical_evaluation}
\end{figure*}

\subsection{Neighborhood Frequencies for Small-scale Networks}
\label{subsec:exp_varying_connectivity}
We first explore how the distance between neighborhood frequencies in \sys and a uniform distribution is influenced by the network connectivity (measured by spectral gap), the initial position of peers in the graph, and the execution time of \sys.

\textbf{Setup.}
We generate $d$-regular topologies, run \sys on each topology for a period of $ T $ seconds, and observe the neighborhood of each peer after the experiment ends. 
There are $ {{n - 1} \choose {d} }$ possible neighborhoods for a particular peer. 
For a $d$-regular topology with $ n $ peers, we run each experiment $ \left \lceil {{n - 1} \choose {d}} \times \frac{100}{d} \right \rceil $ times.
This ensures that, on average, each neighborhood is observed \num{100} times.
We fix $ d = 4 $ and $ n = 64 $.

Since we aim to analyze how the spectral gap affects the convergence speed of \sys but cannot directly generate topologies with a particular spectral gap, we do the following.
We first generate $ 10^6 $ $4$-regular topologies with 64 peers, compute for each topology its spectral gap $ \lambda $, and select ten of these topologies such that their spectral gaps somewhat uniformly cover values in the range $ \lambda = [0.07, 0.22] $ (these values are the lowest and highest $ \lambda $ we found throughout topology generation).
A higher spectral gap $ \lambda $ indicates a better-connected topology, and we aim to understand its impacts on the convergence of \sys empirically.
We run \sys on these topologies for $T = 4$ and $ T = 6$ seconds while recording the neighborhood frequencies for \emph{each} peer.
We then run a KS test for each previously received neighborhood distribution.

\textbf{Results.}
Figure~\ref{fig:ks_distances} shows the variance in KS-test distances with a boxplot for each tested topology and for $ T = 4 $ (left) and $ T= 6 $ (right).
Each boxplot consists of 64 values, each associated with the KS-test distance computed using the neighborhood distribution of a peer in the topology.
On the one hand, for the topologies with a small spectral gap (\eg, $\lambda = 0.07$ or $ \lambda = 0.088 $), we observe relatively high KS-test distances for $ T = 4$ and small $p$-values ($ p \leq 0.05 $) for individual KS tests, meaning that for these topologies the neighborhood frequencies cannot be reasonably assumed to have converged to a uniform distribution yet\footnote{For the topology with $ \lambda = 0.088 $, we find a disproportionate median KS-test distance of 0.15 with $ T = 4 $. For presentation clarity, we omitted the corresponding boxplot from Figure~\ref{fig:ks_distances}.}.
On the other hand, for topologies with $ \lambda \geq 0.18 $, we find large $p$-values, suggesting that \sys on these topologies has converged at $ T = 4 $.
Most of the KS-test distances are relatively low (below 0.01).
The KS-test distances corresponding to $ T = 6 $ are lower compared to $ T = 4 $, and we find that 89.9\% of all $p$-values produced by the KS test are higher than 0.05.
Furthermore, we observe a low variance of KS-test distances for a particular topology and $ T = 6 $, indicating that the initial starting position of the peer in the graph does not affect its convergence much.
This suggests that the neighborhoods of peers indeed approach a uniform sample over time.

\textbf{Conclusion.} These experiments demonstrate that \sys requires less time to generate uniformly random peer samples on networks with better connectivity (\ie, higher spectral gaps).
For most $4$-regular topologies with $2^6=64$ peers, and irrespective of the initial starting position of peers, the neighborhood distribution closely approximates uniformity after just 6 seconds.
This motivates us to focus on specific topologies and observed peers in further experiments.

\subsection{Peer Frequencies for Large-scale Networks}
\label{subsec:exp_convergence_large_topologies}
Since ensuring coverage of all unique neighborhoods quickly becomes an intractable problem when the network size increases, we have limited our empirical analysis in Section~\ref{subsec:exp_varying_connectivity} to relatively small topologies ($ n = 64 $).
To evaluate \sys with larger topologies, we count the frequencies on a per-peer instead of a per-neighborhood basis. %In this part, we explore convergence guarantees of \sys on large-scale networks. 

\begin{figure*}[t]
\vspace{-0.3cm}
\centering
\includegraphics[width=\linewidth]{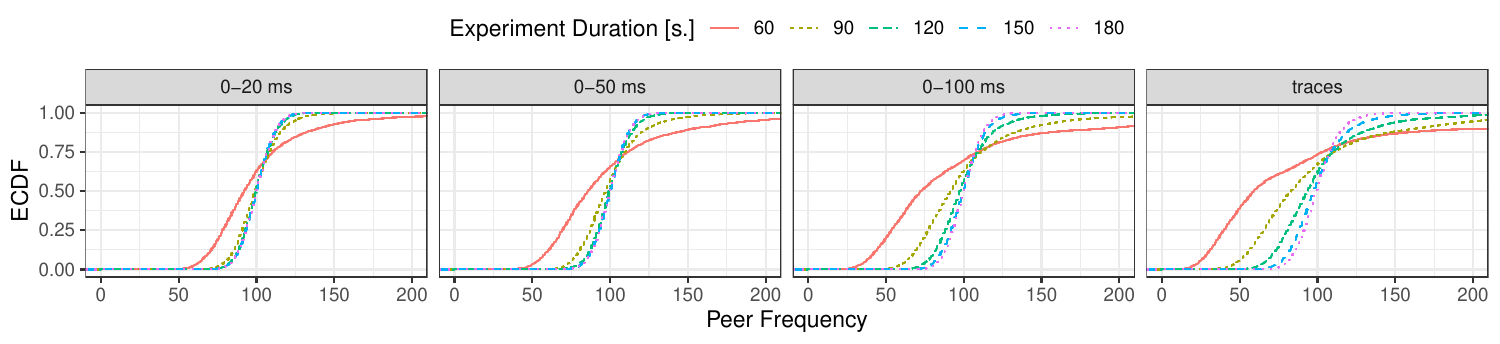}
\caption{Distribution of peer frequencies for different $ \delta_{max} \in \{20, 50, 100\} $ values (in milliseconds) and with realistic traces applied (right-most plot), for different experiment durations.}
\label{fig:nodes_frequencies_different_t_latencies_ecdf}
\vspace{-0.3cm}
\end{figure*}

\textbf{Setup.}
Similar to the previous experiments, we repeat each experiment until, on average, each peer is observed \num{100} times.
Also, based on the outcome of our previous experiment in Section~\ref{subsec:exp_varying_connectivity}, we established that results do not vary significantly, either for different observed peers or for topologies with different but relatively high connectivity.
Therefore, in this set of experiments, we average the results for each pair of $(n, d)$ over five different random topologies while examining one random peer.
We focus on how KS-test distances change with different network sizes $n$ and durations of the \sys execution $T$.

\textbf{Results.}
Figure~\ref{fig:nodes_frequencies_different_t_ecdf} shows the empirical cumulative distribution function (ECDF) of peer frequencies for topologies with $ n = \num{4096} $, $ d = 5$ and with five different experiment durations $ T $. We visually notice that the distribution of peer frequencies quickly approaches uniform sampling, \eg, peer frequencies become more concentrated around \num{100}. Figure~\ref{fig:nodes_frequencies_different_n_ecdf} shows the distribution of peer frequencies for increasing values of $ n $ up to $ \num{32768} $ ($ 2^{15} $) peers while fixing $ d = 5 $ and $ T = 5 $.
We observe that all network sizes yield peer frequencies visually indistinguishable from one another and from those obtained through uniform sampling. Finally, we summarize the KS-test distances for each value of $ n $ and topology in Figure~\ref{tab:ks_test_large_topologies}. We can see that the KS-test distance ranges from 0.003 to 0.02 and provides large $p$-values for all tests.

\textbf{Conclusion. } For the evaluated large-scale topologies, the distribution of peer frequencies approaches that of uniform sampling in a short time.
This demonstrates the convergence of \sys in large networks.

\begin{figure}[b]
\centering
\includegraphics[width=\linewidth]{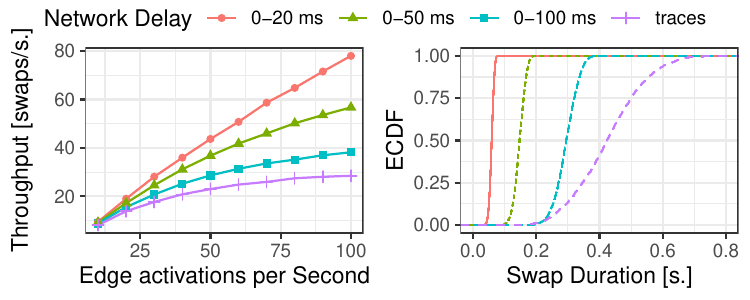}
\caption{The throughput of \sys (in swaps/s.) for increasing frequencies of edge activation (left) and the distribution of individual swap durations (right), for different network delays.}
\label{fig:peerswap_with_network_delays}
\end{figure}

\subsection{Convergence and Throughput of \sys with Network Delays}
\label{subsec:convergence_peerswap_with_network_delays}
In this part, we explore the convergence and performance of \sys in the presence of network delays. 
We use the adapted protocol as described in Section~\ref{sec:peerswap_with_delays}. 
Our goal is to see if the distribution of peer samples provided by \sys converges to the uniform one even if messages are delayed.
Additionally, we believe that adjusting the main parameter of \sys, the Poisson clocks rate $\alpha$, can result in more swaps per second and thus help to improve the convergence speed. 
To explore this, we specifically focus on the average throughput of \sys, \ie, the average number of (successful) swaps per second, for varying network delays and edge activation frequencies.

\textbf{Setup.}
We consider a network of $ n = 1024 $ peers, connected in a $d$-regular topology with $ d = 5 $.
For each pair of peers, we generate pair-wise network delays through uniformly drawing a delay between the interval $ [0, \delta_{max}] $.
This set of experiments considers three synthetic delay values -- those of $ \delta_{max} \in \{20, 50, 100\}$ ms.
We also run a fourth configuration with realistic network traces that are sourced from WonderNetwork~\cite{WonderNetwork}.
For the sake of presentation, we present results with respect to the average number of edge activations in the system, denoted by $r$. % -- the amount of swap attempts per second.  
This value depends on the total number of edges in the graph and Poisson clock rate $\alpha$.

First, we run a similar experiment to Section~\ref{subsec:exp_convergence_large_topologies} in the presence of delays. We consider various durations of $ T \in\{60, 90, 120, 150, 180 \}$ seconds and derive the distribution of peer frequencies provided by samples of \sys.
Here, we adjust the Poisson clock rate $\alpha$ in such a way that on average \num{50} swaps per second are initiated.

Second, we control the interconnected values of $\alpha$ and $r$, and estimate the influence of Poisson clocks' rates on the \sys throughput.
We vary $r$ from \num{10} to \num{100} edge activations per second. Our goal is to estimate 
\sys throughput, defined as the average number of successful swaps per second. Higher is the throughput, quicker is the convergence of \sys.
We run each experiment for two minutes and repeat for five seeds.

\textbf{Convergence results.}
Figure~\ref{fig:nodes_frequencies_different_t_latencies_ecdf} shows the empirical cumulative distribution function (ECDF) of peer frequencies with varying delays $ \delta_{max}$ and when using traces.
Comparing the plots in Figure~\ref{fig:nodes_frequencies_different_t_latencies_ecdf}, we can see that it takes longer for \sys to converge when the larger delays or realistic traces are applied.
This is due to the lower throughput of \sys with realistic traces and large delays than with smaller delays.

\textbf{Throughput results.}
Figure~\ref{fig:peerswap_with_network_delays} (left) shows the throughput of \sys, for different values of $ \delta_{max} $ and with our realistic traces.
Increasing the swap frequency generally increases the achievable throughput.
We find this increase to flatten out, which is particularly noticeable for the experiments with larger network delays.
This can be explained by two reasons. 
First, the probability of two \enquote{conflicting} swaps being initiated at the same time increases as more swaps are initiated within the same time period.
Second, when network delays increase, each swap takes longer to finish, resulting in prolonged periods of nodes being locked.
Figure~\ref{fig:peerswap_with_network_delays} (right) shows the distribution of swap durations for each network delay profile, indeed confirming our observation. 
While throughput optimization is not the main focus of our paper, modifying the Poisson clock rate is a tangible way of improving in performance, especially if message delays are mild (\ie, up to 20 ms).
This comes, however, at the cost of additional network traffic.

\textbf{Conclusion. }
Our empirical results hint that \sys can reach quick convergence even in challenging network conditions, \eg, when deployed on the Internet.
Also, adjustment of the Poisson clocks rate can significantly improve the throughput and protocol convergence.

\section{Concluding remarks}
\label{section:future}
Our work makes the first step towards establishing gossip-based peer samplers with provable randomness guarantees.
We introduced \sys, a protocol which provides each peer with a random sample within a time frame determined by the network's size and connectivity.
The idea underlying \sys is simple:  two adjacent peers periodically swap their entire neighborhood. The theoretical analysis, which is more challenging, has been made possible by establishing a link between the dynamics of peers in \sys and interchange processes.
We also introduced a variant of \sys that operates in the presence of network delays.
Our numerical evaluation using a simulated version of \sys conveys the convergence of \sys over time, thus confirming that it is an efficient and scalable peer sampler.

\section*{Acknowledgment} 
The authors are thankful to Roberto Oliveira for his insights about interchange processes and Laurent Massoulie, Nirupam Gupta, Akash  Dhasade and  Rishi Sharma for fruitful discussions and proofreading the paper. 

\bibliographystyle{acm}
\bibliography{biblio}

\section{Postponed Proofs}
\label{appendix}
\subsection{Additional preliminaries}

\textbf{Random walks}
A \emph{random walk} $RW(G)$ on a graph $G=(V, E)$ is a ctMC with state space $V$ that describes one peer randomly moving from node to node in a graph $G$, following the existing edges of $G$. Formally as per \cite{oliveira2013mixing}, random walk can be described formally by its infinitesimal generator using transposition function \ref{eq:transition_f-1}. For any distinct $u, v\in V$:
\begin{equation}
    q(u, v) = 
    \begin{cases}
        1,&\text{if $\exists e\in E$ s.t. $f_e(u)=v$,}\\
        0,&\text{otherwise. }
    \end{cases}
\end{equation}

\textbf{Mixing times of RW and IP.} The mixing time of a $RW(G)$ can be bounded using a  measure of the graph connectivity, its spectral gap $\lambda(G)$. 

\begin{lemma}
\label{lemma:rw-mixing-time}
Consider a continuous-time random walk $X$ on a $d$-regular graph $G=(V, E)$, $|V|=n$. Then, for any $\varepsilon >0$ the $\varepsilon$-mixing time of $X$ can be bounded as follows:
    \[T_X(\varepsilon)\leq \log\left(\frac{n}{\varepsilon }\right)\frac{1}{d\cdot \lambda(G)}.\]
\end{lemma}

\begin{proof}
Theorem 20.6~\cite{levin2017markov} proves that the mixing time of any ctMC $X$ with set of states $S$ and stationary distribution $\pi=(\pi_i)_{i\in S}$ can be bounded as follows:
        \begin{equation}\label{eq:mixing}
        T_{X}(\varepsilon)\leq \log\left(\frac{1}{\varepsilon \min\limits_{i\in S}(\pi_i)}\right)\frac{1}{r\cdot \lambda(P)},
        \end{equation}
    where $r=\max\limits_{s\in S}\sum\limits_{s'\in (S\backslash \{s\})}q(s, s')$, and $P=Q/r+I$. When $X$ is a random walk, we can calculate all these values. Specifically \footnote{Analogue of the discrete variant described in Section 1.5~\cite{levin2017markov}}, $\pi_i = \frac{\deg(i)}{2|E|}$ for $\forall i\in V$, which in case of $d$-regular graph on $n$ nodes gives  $\min\limits_{i\in S}(\pi_i)=\frac{1}{n}$. Furthermore, for any $i\in V$ $\sum\limits_{j\in (V\backslash \{i\})}q(i, j)=d$ and thus $r=d$.  We also note that $Q/r+I=W_G$, hence the spectral gap of $P = W_G$. Substituting these values in \eqref{eq:mixing} concludes the lemma. 
\end{proof}

\begin{remark}
    A reader may notice that, compared to a discrete random walk~\cite{levin2017markov}[Theorem 12.4], the mixing time for the continuous analogue is reduced by a factor of $d$, the degree of the graph. This reduction occurs because, in the continuous-time random walk, the transitions between nodes happen $d$ times faster, even though the total number of transitions required to achieve mixing remains the same.
\end{remark}

Mixing time of an interchange process can be bounded by the mixing time of a random walk on the same graph. 

\begin{theorem}[Theorem 1.1 from ~\cite{oliveira2013mixing}]\label{thm:interchange-rw} 
    There exists a universal constant $C>0$ such that for all $\varepsilon\in (0, 1/2)$, all connected graphs $G = (V,E)$, with $|V|\geq 2$, and for all $k\in\{1, \dots, |V|-1\}$, if we consider $X$ to be a $RW(G)$ and $X'$ to be an $IP(k, G)$, then
    \[T_{X'}(\varepsilon) \leq C\ln \left( \frac{|V|}{\varepsilon} \right) T_{X}(1/4).\]
\end{theorem}

\subsection{Proofs of Section~\ref{section:interchange}}
\label{appendix:proofs1}

\FirstBijection*
\begin{proof}
   
    Let $t\geq0$ and denote by $N\geq 0$ the number of times the Poisson clocks rang before time $t$.
    If $N=0$, we set $\sigma_t$ to be the identity. Otherwise, we denote the times at which any of the clocks rang by $t_1, t_2, \dots, t_N$ with $0=t_0<t_1<\dots<t_N\leq t$ and we denote $e_m=(i_m, j_m)\in E(t_{m-1})$ the edge on which the clock rang at time $t_m$. Then we set 
    \begin{equation}\label{eq:bijection}
     \sigma_t = f_{e_N} \circ \dots \circ f_{e_1}
    \end{equation}
    where for every $m\in \{1, \dots ,N\}$, $f_{e_m}$ is as defined in \eqref{eq:transition_f-1}.
    Note that for any $N \geq 1$, $m \in \{1, \dots, N\}$ and $e_m \in E(t_{m-1})$, $f_{e_m}$ is a bijection, hence $\sigma_t$ is also a bijection (as a composition of bijections). Furthermore, by construction we have that $(i,j) \in E(t)$ if and only if $(\sigma_t^{-1}(i), \sigma_t^{-1}(j))\in E(0)$ for any $t\geq 0$. Hence, we finally get that $C_i(j)=C_e$ with $e=(\sigma_t^{-1}(i),\sigma_t^{-1}(j)) \in E(0)$.

\end{proof}

\MainBijection*

\begin{proof}
Since $\sigma_t = \sigma_{t_N}$ and $\gamma_t=\gamma_{t_N}$, it is sufficient to prove the lemma statement for $t_N$. We will do it by induction over $N$. We know that $\eh_1=e_1\in E(0)$ and by construction of $f_e$ for any $e\in E(0)$ holds $f_e=f^{-1}_e$. Then for $N=1$:
\[\sigma^{-1}_{t_1} = f^{-1}_{e_1} = f_{e_1} = f_{\eh_1}= \gamma_{t_1}\]
and the lemma statement holds. Let $M \geq 1$ be an arbitrary value. We assume that for all $N\leq M$ the lemma holds. We show below that, in that case, it also holds for $N=M+1$. 

First recall that  $\sigma_{t_{M+1}}= f_{e_{M+1}} \circ \dots \circ f_{e_1}$, and by construction of $f_e$ for any $e_m\in E(t_{m})$, $f_{e_m}=f^{-1}_{e_m}$. Then 
\begin{align*}
    \sigma^{-1}_{t_{M+1}}= f^{-1}_{e_1} \circ \dots \circ f^{-1}_{e_{M+1}}= f_{e_1} \circ \dots \circ f_{e_{M+1}}.
\end{align*}
      
Thus $\sigma^{-1}_{t_{M+1}} = \sigma^{-1}_{t_M}\circ f_{e_{M+1}}$ with $e_{M+1}=(i_{M+1}, j_{M+1})\in E(t_{M})$. Therefore, for any $v\in V$
      \begin{equation}\label{eq:transition_sigma}
        \sigma^{-1}_{t_{M+1}}(v) = \begin{cases}
            \sigma^{-1}_{t_{M}}(v),&\text{if $v\notin \{i_{M+1}, j_{M+1}\}$,}\\
            \sigma^{-1}_{t_{M}}(i_{M+1}),&\text{if $v=j_{M+1}$,}\\
            \sigma^{-1}_{t_{M}}(j_{M+1}),&\text{if $v=i_{M+1}$.}\end{cases}
        \end{equation}
      We want to prove that $\sigma^{-1}_{t_{M+1}} = \gamma_{t_{M+1}}$. When we decompose $ \gamma_{t_{M+1}}$ and use induction assumption that $\sigma^{-1}_{t_{M}} = \gamma_{t_{M}}$, we get:
      \begin{align*}
        \gamma_{t_{M+1}} &= f_{\eh_{M+1}} \circ \dots \circ f_{\eh_1} = f_{\eh_{M+1}}\circ \gamma_{t_M} =f_{\eh_{M+1}}\circ \sigma^{-1}_{t_M}, 
      \end{align*} 
      where 
      \begin{equation}\label{eq:transition_sigma_2}
          \eh_{M+1} = (\ih_{M+1}, \jh_{M+1}) = (\sigma^{-1}_{t_{M}}(i_{M+1}), \sigma^{-1}_{t_{M}}(j_{M+1})).
      \end{equation} By definition, $f_{\eh_{M+1}}$ changes only those values that are neither $\ih_{M+1}$ nor $\jh_{M+1}$.
      Then 
      \begin{equation}\label{eq:transition_sigma_3}
        \gamma_{t_{M+1}}(v) = \begin{cases}
            \sigma^{-1}_{t_{M}}(v),&\text{if $\sigma^{-1}_{t_{M}}(v)\notin\{\ih_{M+1}, \jh_{M+1}\}$,}\\
            \jh_{M+1},&\text{if $\sigma^{-1}_{t_{M}}(v)=\ih_{M+1}$,}\\
            \ih_{M+1},&\text{if $\sigma^{-1}_{t_{M}}(v)=\jh_{M+1}$.}\end{cases}
        \end{equation}
        Substituting  $\ih_{M+1}, \jh_{M+1}$ from \eqref{eq:transition_sigma_2} into \eqref{eq:transition_sigma_3}, and noting that (since the mapping is a bijection) $v=u$ if and only if $\sigma_{t_M}^{-1}(v)=\sigma_{t_M}^{-1}(u)$ shows that \eqref{eq:transition_sigma} and \eqref{eq:transition_sigma_3} are equivalent. This concludes the lemma. 
\end{proof}

\begin{lemma}\label{lemma:time-properties}
     Consider Poisson clocks $(C_e)_{e\in E}$ assigned to edges in \sys process on the network $G=(V, E)$. Let Poisson ring times on each edge $e\in E$ be $(T^e_n)_{n \geq 1}$. For any $0\leq t_1<t_2$, denote by $T^e[t_1, t_2]$ a sequence of time period values $t'\leq t_2-t_1$ such that $C_e$ ringed at time $t_1+t'$ within time period $[t_1,t_2]$, i.e. $T^e[t_1, t_2]:=\{t': 0\leq t'< t_2-t_1, t'+t_1\in (T^e_n)_{n \geq 1} \}$. Denote $T[t_1, t_2]:=(T^e[t_1, t_2])_{e\in E}$, $0\leq t_1<t_2$. Then, for any $0 \leq t_1< t_2<t_3$
     \begin{itemize}
         \item $T[t_1, t_2]$ and $T[t_2, t_3]$ are independent;
         \item $T[t_2, t_3]$ is time-homogeneous.
    \end{itemize}
\end{lemma}

\begin{proof}
    First, we show the lemma statement for $T^e, \forall e\in E$, that for any $0 \leq t_1< t_2<t_3$ holds that $T^e[t_2, t_3]$ and $T^e[t_1, t_2]$ are time-homogeneous and independent from each other. The random variable measuring the length of the interval from $t_2$ until the first ring time of $C_e$ after $t_2$ is a non-negative random variable $\hat{t}$ with exponential distribution $\mathcal{E}(1)$. $\hat{t}$ is independent of all ring times before $t_2$ (Theorem 2.2.1~\cite{poisson_processess}), and thus independent of $T^e[t_1, t_2]$.
    According to the definition of Poisson Clock, ring times on edge $e$ are defined by a sequence of independent exponential random variables $(Z_k)_{k\geq 0}$ and $t_1+\hat{t}=\sum\limits_{k\in[1, l]}Z_k$ for some $l$. Then, next values in $T^e[t_2, t_3]$ are of form
    \begin{equation}\label{eq:time1}
        \hat{t}+\sum\limits_{k\in[l+1, m]}Z_k\leq t_3-t_2.
    \end{equation} Since $\hat{t}$ and $(Z_k)_{k\geq l}$ are independent of $T^e[t_1, t_2]$, then $T^e[t_2, t_3]$ is independent with $T^e[t_1, t_2]$. Moreover, $\hat{t}$ and $(Z_k)_{k\geq l}$ are also independent of the exact value $t_2$ and have only restriction \eqref{eq:time1} which is defined by the length of time period $t_3-t_2$. Thus, $T^e[t_2, t_3]$ is time-homogeneous. 
    Finally, since for all $0<t'< t''$ all $T^e[t', t''], \forall e\in E$ are independent, then the lemma statement also holds for their junction $T[t', t'']$. 
\end{proof}

\begin{theorem}
    The process $\swap^{(k)}$ is Markovian and time-homogeneous. 
\end{theorem}

\begin{proof}
Consider any $0<t'<t''$ and let $e_N, e_{N+1}, ... e_M$ be the list of edges on which Poisson clocks ringed at period $[t', t'']$. Then $\swap^{(k)}_{t''}=\sigma_{t''}(\swap^{(k)}_0)=f_{e_M}\circ ... \circ f_{e_N}\circ \swap^{(k)}_{t'}$ and also $\swap^{(k)}_{t'}=f^{-1}_{e_N}\circ ... \circ f^{-1}_{e_M}\circ \swap^{(k)}_{t''}$. Note, that we can derive sequence  $e_N, e_{N+1}, ... e_M$ from $T[t', t'']$. Thus, given $T[t', t'']$ we can uniquely define both $\swap^{(k)}_{t''}$ from $\swap^{(k)}_{t'}$  and vice versa, \ie, exists some function $F$ such that
\begin{align}\label{eq:mark:1}
    F(\swap^{(k)}_{t'},T[t', t'']) = \swap^{(k)}_{t''}, F^{-1}(\swap^{(k)}_{t''},T[t', t'']) = \swap^{(k)}_{t'}.
\end{align}

Using these results, we will prove the Markovian property
\begin{align}\label{eq:markovian}
    \Pr[\swap^{(k)}_{t_n}=s_n\mid \swap^{(k)}_{t_1}=s_1, ..., \swap^{(k)}_{t_{n-1}}=s_{n-1}]=
    \\\Pr[\swap^{(k)}_{t_n}=s_n\mid \swap^{(k)}_{t_{n-1}}=s_{n-1}]\nonumber
\end{align}
for all $s_1, ...s_k\in (V)_k$ and any sequence $0\leq t_1\leq t_2...\leq t_n$ of times. $\swap^{(k)}_{t_{n}}$ is uniquely defined by $\swap^{(k)}_{t_{n-1}}$ and $T[t_{n-1}, t_{n}]$, and $\swap^{(k)}_{t_{1}},...,  \swap^{(k)}_{t_{n-2}}$ are uniquely defined by $\swap^{(k)}_{t_{n-1}}$ and $T[t_1, t_{n-1}]$. Also, by Lemma~\ref{lemma:time-properties}, $T[t_1, t_{n-1}]$ and $T[t_{n-1}, t_{n}]$ are independent. Thus $(\swap^{(k)}_{t_{n}}\mid\swap^{(k)}_{t_{n-1}})$ is independent of $(\swap^{(k)}_{t_{1}}, ..., \swap^{(k)}_{t_{n-2}})$. This concludes \eqref{eq:markovian} and that the process is Markovian. 

Moreover, since $T$ is time-homogeneous and $T[t_{n-1}, t_n]$ together with value $Y^{(k)}_{t_{n-1}}$ fully defines value of $Y^{(k)}_{t_n}$, $Y^{(k)}$ is also time-homogeneous. Formally, we describe it using the function $F$ from \eqref{eq:mark:1}:
\begin{align*}
\Pr[\swap^{(k)}_{t_n}=s_n&\mid \swap^{(k)}_{t_{n-1}}=s_{n-1}]=\Pr[F(s_{n-1}, T[t_{n-1}, t_{n}])=s_{n}] \\ &=\Pr[F(s_{n-1}, T[0, t_{n}-t_{n-1}])=s_{n}] \\
&= \Pr[\swap^{(k)}_{t_n-t_{n-1}}=s_n\mid \swap^{(k)}_{0}=s_{n-1}],    
\end{align*}
which concludes the proof. 
\end{proof}

\SwapIsInterchange* 

\begin{proof}
    Since the process $\swap^{(k)}$ is a time-homogeneous ctMC, it can be uniquely defined by its infinitesimal generator. Thus, it is sufficient to compute $q({\mathbf{x}, \mathbf{y}})$ for any two $\mathbf{x}, \mathbf{y}\in (V)_k$ and show that $q({\mathbf{x}, \mathbf{y}})$ is identical to that of an interchange process (see \eqref{eq:interchange_transition}).  Since $q(\mathbf{x}, \mathbf{x})=-\sum\limits_{\mathbf{y}\neq \mathbf{x}}q(\mathbf{x}, \mathbf{y})$ for any $\mathbf{x} \in (V)_k$, below we only consider cases where $\mathbf{x}\neq \mathbf{y}$. By definition, for $\mathbf{x}\neq \mathbf{y}$ we have
   $ q(\mathbf{x}, \mathbf{y}) = 
    \lim\limits_{h\rightarrow 0+}\frac{P_{\mathbf{x}, \mathbf{y}}(h)}{h}$. Hence, we need to compute $P_{\mathbf{x}, \mathbf{y}}(h)$ and $\lim\limits_{h\rightarrow 0+}\frac{P_{\mathbf{x}, \mathbf{y}}(h)}{h}$. To do so, let us consider the following events:
\vspace{-5pt}
\begin{align*}
\Gamma_0 &= \{\text{No Poisson clock rings within $[0, h]$}\},\\
\Gamma_e &= \{\text{One Poisson clock rings within $[0, h]$, at $e\in E(0)$}\},\\
\Gamma_2 &= \{\text{Two or more Poisson clocks ring within $[0, h]$}\}.
\end{align*}
Since $\Gamma_0, \{\Gamma_e\}_{e\in E(0)}, \Gamma_2$ create the disjoint partition of the space, we can use the total probability formula:
\begin{align}\label{eq:prob_decomposition}
    P_{\mathbf{x}, \mathbf{y}}(h) = &~\Pr[Y^{(k)}_h=\mathbf{y}\mid Y^{(k)}_0=\mathbf{x}, \Gamma_0]\cdot \Pr[\Gamma_0] \\ 
    &+ \sum \limits_{e\in E}\left(\Pr[Y^{(k)}_h=\mathbf{y}\mid Y^{(k)}_0=\mathbf{x}, \Gamma_e]\cdot \Pr[\Gamma_{e}] \right)\nonumber\\
    &+ \Pr[Y^{(k)}_h=\mathbf{y}\mid Y^{(k)}_0=\mathbf{x}, \Gamma_2]\cdot \Pr[\Gamma_2] \nonumber
\end{align}
Now, we compute parts of this expression. We denote the total amount of edges to be $|E(0)|=\mathfrak{m}$ and $\e$ to be Euler's constant. Note that the Poisson clocks on each edge ring after time periods that have exponential distribution $\mathcal{E}(1)$ and that Poisson clocks are independent of each other. Thus, $\Pr[\Gamma_0]$ is the joint probability of all the Poisson clocks having their first ring after the time period $[0, h]$ and $\Pr[\Gamma_e]$ is the probability that only one of them (i.e, $e$) rings at least once before time $h$. Hence, we have for any $e \in E(0)$
\begin{align}
    &\Pr[\Gamma_0] = (\e^{-h})^{\mathfrak{m}}=\e^{-h\mathfrak{m}} \label{eq:proba_0}\\
    &\Pr[\Gamma_e] = (1-\e^{-h})\e^{-h(\mathfrak{m}-1)}. \label{eq:proba_e}
\end{align}
Furthermore, using Taylor series at $h\rightarrow 0+$, $\forall e \in E$
\begin{align*}
    \lim \limits_{h\rightarrow 0+ }\frac{\Pr[\Gamma_e]}{h}&=
    \lim \limits_{h\rightarrow 0+ }\tfrac{\sum\limits_{l\geq1}\left(-\frac{(-h)^l}{l!}\right)\e^{-h(\mathfrak{m}-1)}}{h}\\
    &=\lim \limits_{h\rightarrow 0+ }\sum\limits_{l\geq0}\left(\frac{(-h)^{l}}{(l+1)!}\right)\e^{-h(\mathfrak{m}-1)}=1.
\end{align*}

Now, we compute $\Pr[\Gamma_2]$:
\begin{align*}
   &\Pr[\Gamma_2] = 1- \Pr[\Gamma_0]-\mathfrak{m}\cdot\Pr[\Gamma_e] \\
   &=1-\e^{-h \mathfrak{m}}-\mathfrak{m}(1-\e^{-h})(\e^{-h(\mathfrak{m}-1)})\\
   &=(1-\e^{-h})\left(\sum\limits_{i\in\{0,\dots,  \mathfrak{m}-1\}} \hspace{-10 pt} \e^{-i\cdot h}-\mathfrak{m}\e^{-h(\mathfrak{m}-1)}\right).
\end{align*}
Again, by using Taylor series we can rewrite the limit $\lim\limits_{h\rightarrow 0+ } \Pr[\Gamma_2]/h$ as 
\begin{equation*}
    \lim \limits_{h\rightarrow 0+ }\sum\limits_{l\geq0}\frac{(-h)^{l}}{(l+1)!}\left(\sum\limits_{i\in\{0,\dots, \mathfrak{m}-1\}}\hspace{-10pt} \e^{-i\cdot h}-\mathfrak{m}\e^{-h(\mathfrak{m}-1)}\right).
\end{equation*}
Since the first term tends to $1$ when $h\rightarrow 0+$, and the second term tends to $0$, we get $\lim \limits_{h\rightarrow 0+ }\Pr[\Gamma_2]/h=0.$ 

To finish calculation of $q(\mathbf{x},\mathbf{y}), \forall \mathbf{x},\mathbf{y}\in (V)_k$, we study the two following cases. %states $\mathbf{x}$ and $\mathbf{y}$.

\paragraph{Case 1. $\mathbf{x}\neq \mathbf{y}$, and there exists an edge $e'$ such that $f_{e'}(\mathbf{x})=\mathbf{y}$}
Then by definition
\begin{align*}
    \Pr[Y^{(k)}_h=\mathbf{y}\mid Y^{(k)}_0=\mathbf{x}, \Gamma_e]=
    \begin{cases}
        0,&\text{if $e\neq e'$,}\\
        1,&\text{if $e=e'$. }
    \end{cases} 
\end{align*}
 Also, $\Pr[Y^{(k)}_h=\mathbf{y}\mid Y^{(k)}_0=\mathbf{x}, \Gamma_0]=0$ and $\Pr[Y^{(k)}_h=\mathbf{y}\mid Y^{(k)}_0=\mathbf{x}, \Gamma_2]\leq 1$. Substituting these and \eqref{eq:prob_decomposition} in the definition of $ q(\mathbf{x}, \mathbf{y})$, we obtain that in this case 
\begin{align*}
     q(\mathbf{x}, \mathbf{y}) &= \lim\limits_{h\rightarrow 0+}\frac{P_{\mathbf{x}, \mathbf{y}}(h)}{h}\\
     &=\lim\limits_{h\rightarrow 0+}\frac{\Pr[Y^{(k)}_h=\mathbf{y}\mid Y^{(k)}_0=\mathbf{x}, \Gamma_{e'}]\cdot \Pr[\Gamma_{e'}]}{h}=1.
\end{align*}

\paragraph{Case 2. $\mathbf{x}\neq \mathbf{y}$, and there is no edge $e'$ such that $f_{e'}(\mathbf{x})=\mathbf{y}$} 
In this case, for any $e\in E(0)$ we have \[\Pr[Y^{(k)}_h=\mathbf{y}\mid Y^{(k)}_0=\mathbf{x}, \Gamma_e]=0,\] 
\[\Pr[Y^{(k)}_h=\mathbf{y}\mid Y^{(k)}_0=\mathbf{x}, \Gamma_0]=0.\] Also, $\Pr[Y^{(k)}_h=\mathbf{y}\mid Y^{(k)}_0=\mathbf{x}, \Gamma_2]\leq 1$. Substituting these and \eqref{eq:prob_decomposition} in formula of $ q(\mathbf{x}, \mathbf{y})$, we get in this case  
\begin{align*}
     q(\mathbf{x}, \mathbf{y}) = \lim\limits_{h\rightarrow 0+}\frac{P_{\mathbf{x}, \mathbf{y}}( h)}{h}=0.
\end{align*}
These two cases show that the infinitesimal generator of the process $\swap^{(k)}$ and that of an interchange process  $IP(k, G(0))$(as defined in \eqref{eq:interchange_transition}) are the same.

\end{proof}

\subsection{Proofs of Section~\ref{section:mixing-time}}
\label{appendix:proofs2}

\FirstConvergenceTheorem*

\begin{proof}
For the sake of brevity, we prove the theorem statement for the default value of Poisson clocks rate $\alpha=1$ and in the end we just scale the final time value by $\alpha$. 

Let us consider any peer $u_i\in U$ and an ordered tuple of $d$ peers $\mathbf{u}=(u_{i_1}, u_{i_2}, \dots , u_{i_d}) \in (U\backslash\{u_i\})_d$. Notice that  $Sample_i(d, T)$ is a random of $d!$ orderings of the neighborhood of peer $u_i$ % -- $(u_j)_{j\in N(i)}$ 
in graph $G(T)$. Therefore it holds true that:
\begin{equation} \label{eq:f0}
\Pr[Sample_i(d, T)=\mathbf{u}] = \frac{1}{d!}\Pr[N(u_i, T)=\{u_{i_1}, u_{i_2}, \dots , u_{i_d}\}].
\end{equation}

  Consider an interchange process $Y^{(d+1)}=IP(d+1, G(0))$ with $Y^{(d+1)}_0=(i, i_1, i_2, \dots, i_d)$, the initial positions of peers $(u_i, u_{i_1}, \dots, u_{i_d})$ in $G(0)$. For the sake of brevity we remove upper index in $Y^{(d+1)}$ and refer to this process $Y=(Y_t)_{t\geq 0}$.  We recall that set of states of $Y$ is $(V)_{d+1}$. Then, thanks to Theorem \ref{thm:swap-is-interchange}, characterizing the connections between $(u_i, u_{i_1}, \dots, u_{i_d})$ in $(G(t))_{t\geq 0}$ (when studying the \sys process\footnote{Note that, when studying \sys each peer is associated to a single node. Hence studying the connections between peers simply means tracking $E(t)$.}) is equivalent to characterizing the connections (according to $G(0)$) between the nodes in the states of $Y$. Thus, 
\begin{equation} \label{eq:f1}
\Pr[N(u_i, T)=\{u_{i_1}, u_{i_2}, \dots , u_{i_d}\}] = \hspace{-25pt}\sum\limits_{\substack{\mathbf{v}\in (V)_{d+1}:\\N(\mathbf{v}[1], 0) = \{\mathbf{v}[2], \dots, \mathbf{v}[d+1]\}}}\hspace{-25pt}\Pr[Y_T = \mathbf{v}].
\end{equation}
The sum on the right calculates the probability that particles of process $Y$ at time $T$ are positioned on set of nodes where the first node has others as a neighborhood.  We will estimate this sum now. Since the graph $G(0)$ is connected, there exists a unique stationary distribution $\pi$ of the process $Y$ ~\cite{levin2017markov}[Example 1.12.]. Then by definition of $\varepsilon$-mixing time for $T=T_Y(\varepsilon)$, the following holds: 
\begin{align*}
   2\varepsilon &\geq 2 d_{TV}(P_{Y_0,\cdot }(T), \pi)
   = \sum\limits_{\mathbf{v}\in (V)_{d+1}}\hspace{-13pt}\mid\Pr \left[Y_T = \mathbf{v} \right]
    -\pi\left (\mathbf{v}\right ) \mid
\end{align*}
Since this sum is bigger than its part, and the  sum of modulos is bigger then the modulo of the sum, this gives us
\begin{align}\label{eq:f2}
    2\varepsilon \geq \left|\sum\limits_{\substack{\mathbf{v}\in (V)_{d+1}:\\N(\mathbf{v}[1], 0) = \{\mathbf{v}[2], \dots, \mathbf{v}[d+1]\}}}\hspace{-30pt}\left( \Pr[Y_T=\mathbf{v}] - \pi(\mathbf{v}) \right) \right|
\end{align}
For the sake of brevity we denote $\rho := \frac{1}{d!}\hspace{-20pt}\sum\limits_{\substack{\mathbf{v}\in (V)_{d+1}:\\N(\mathbf{v}[1], 0) = \{\mathbf{v}[2], \dots, \mathbf{v}[d+1]\}}} \hspace{-20pt}\pi( \mathbf{v})$, and apply \eqref{eq:f0} and  \eqref{eq:f1} to the result of \eqref{eq:f2} to get
\begin{align}\label{eq:f3}
     2\varepsilon\geq \frac{2\varepsilon}{d!}\geq\left|\Pr[Sample_i(d, T)=\mathbf{u}] - \rho\right|.d
\end{align}

Now, let us consider any other ordered tuple with distinct peers $\mathbf{u'}=(u_{j_1}, u_{j_2}, \dots, u_{j_d}) \in (U\backslash\{u_i\})_d$ who can be neighbors of $u_i$. We define $Y'$ an interchange process $IP(d+1, G(0))$ with initial state $Y'_0 =(i, j_1, j_2, \dots, j_d)$,  that describes the dynamic of $(u_i, u_{j_1}, u_{j_2}, \dots , u_{j_d})$. The stationary distribution and $\varepsilon$-mixing time of an interchange process $IP(d+1, G(0))$ are independent from the initial state, and these values are the same for $Y'$  and $Y$. Since \eqref{eq:f3} holds for any $\mathbf{u}\in (U\backslash\{u_i\})_d$, it also applies to $\mathbf{u'}$ with the same value $\rho$ and $T$. We use \eqref{eq:f3} and triangle inequality to bound difference between probabilities of samples that can be provided by \sys to peer $u_i$: 
\begin{align*} %\label{eq:3}
    &\mid \Pr[Sample_i(d, T)=\mathbf{u'}]- \Pr[Sample_i(d, T)=\mathbf{u}]\mid\leq \\
    & \leq \mid\Pr[Sample_i(d, T)=\mathbf{u'}] -\rho \mid+\\
     &+ \mid\Pr[Sample_i(d, T)=\mathbf{u}] - \rho \mid
     \leq 4\varepsilon. 
\end{align*}

We know that $ \sum\limits_{\mathbf{u}\in (U\backslash\{u_i\})_d} \Pr[Sample_i(d, T)=\mathbf{u}] = 1$. This sum has $\frac{(n-1)!}{(n-d-1)!}$ summands - the number of ways to choose an ordered tuple of $d$ elements out of $n-1$. No two summands can differ by more than $4\varepsilon$. Thus
\begin{align*}
    &\left|\Pr[Sample_i(d, T)=\mathbf{u}]- \frac{(n-d-1)!}{(n-1)!}\right| \leq 4\varepsilon.
\end{align*}

\end{proof}

\FinalConvergenceTheorem*
\begin{proof}
For the sake of brevity, we prove the theorem statement for the default value of Poisson clocks rate $\alpha=1$ and in the end we just scale the final time value by $\alpha$. 

  We consider $T=T_{Y^{(d+1)}}(\varepsilon')$ with $Y^{(d+1)}$ a $IP(d+1, G(0))$ where $\varepsilon'=\frac{\delta}{n^d\cdot 2}$. We apply Theorem \ref{thm:main-1} to each summand of the total variation distance:
    \begin{align*}
        d_{TV}(L(Sample_i(d, T)),Uniform((U\backslash\{u_i\})_d))=\\
      \frac{1}{2}\sum\limits_{\mathbf{u}\in (U\backslash\{u_i\})_d} \hspace{-15pt}|\Pr[Sample_i(d, T)=\mathbf{u}]-\frac{(n-d-1)!}{(n-1)!}|\leq 
      n^d\cdot 2\varepsilon'=\delta. 
    \end{align*}    
Now, to finish the statement of the theorem, we need to estimate $T$. From Theorem \ref{thm:interchange-rw} there exists a universal constant $C>0$ such that for all $\varepsilon\in (0, 1/2)$
    \[T_{IP(d+1, G(0))}(\varepsilon) \leq C\log\left({\frac{n}{\varepsilon}}\right) T_{RW(G(0))}(1/4).\] Using Lemma \ref{lemma:rw-mixing-time}, we bound the mixing time of a random walk:
    \[T_{RW(G(0))}(1/4)\leq \frac{\log(4n)}{d \lambda(G(0))}.\]  
     Therefore, for any $\varepsilon\in(0, 1/2)$
     \[T_{IP(d+1, G(0))}(\varepsilon) \leq C\log\left({\frac{n}{\varepsilon}}\right) \frac{\log(4n)}{d \lambda(G(0))}.\]
Thus, for $\varepsilon'=\frac{\delta}{n^d\cdot 2}$ the following holds
\begin{align*}
    T\leq C\log\left({\frac{n}{\varepsilon'}}\right) \frac{\log(4n)}{d\lambda(G(0))}=
    C\log\left(\frac{2n^{d+1}}{\delta}\right) \frac{\log(4n)}{d\lambda(G(0))},
\end{align*}
which concludes the proof. 
\end{proof}

\section{Decentralized Poisson Clocks}
\label{app:decentralized_poisson_clocks}

In the algorithm description provided in Section~\ref{section:algorithm}, we have assumed that peers have access to a common Poisson clock.
We now explain how one can implement a decentralized Poisson clock.
In practical systems, random values are generated by a pseudorandom number generator (PRNG) defined by some random value distribution and are optionally initialized with a random seed.
We ensure that both peers at the endpoints of a particular edge $ e $ construct a common Poisson clock for $ e $ by initializing a PRNG with a common seed.
We achieve this by having adjacent peers exchange random seeds when the network is initialized.
We assume that all peers in the network are using the same underlying PRNG implementations.

\begin{figure}[t]
\centering
\includegraphics[width=.7\linewidth]{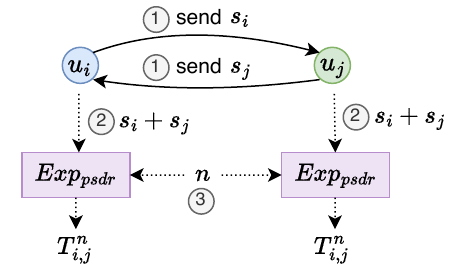}
\caption{\sys realizes a decentralized Poisson clock by having adjacent peers $ u_i $ and $ u_j $ initialize $ Exp_{psdr} $ with a common seed $ s_i + s_j $. Both peers then are able to derive the same list of ring times for the Poisson clock on the edge between $ u_i $ and $ u_j $.}
\label{fig:decentralized_poisson_clock}
\end{figure}

Concretely, our implementation of Poisson clocks without centralized control consists of the following three steps, which are also visualized in Figure~\ref{fig:decentralized_poisson_clock}:
\begin{itemize}
    \item \textbf{Step 1:} When the network is initialized, each peer $u_i$ creates a random seed $s_i$ and sends it to all peers whose identifiers $u_i$ knows $(u_j)_{j\in N(i, 0)}$. Then, $u_j$ awaits the random seeds of these peers (step 1 in Figure~\ref{fig:decentralized_poisson_clock}).
    When $ u_i $ receives a seed $s_j$ of some $u_j, j\in N(i, 0)$, then $u_i$ creates a common seed $s_{i, j}=s_i+s_j$.
    By construction, peer $u_j$ creates the same common seed $s_{j, i}=s_i+s_j$.
    \item \textbf{Step 2:} Peer $ u_i $ then uses the common seed $s_{i, j} $ to initialize a PRNG for values from an exponential distribution with rate $\alpha$ (see definition in Section ~\ref{sec:peer-sampling_systems_model}). We denote this PRNG as $ Exp_{psdr}(\alpha, s_{i, j}) $. Peer $ u_j $ does the same (step 2 in Figure~\ref{fig:decentralized_poisson_clock}).
    \item \textbf{Step 3:} Now that both peers have access to a common PRNG $ Exp_{psdr} $, they are able to generate a list of clock ringing times. By feeding index $ n $ to $ Exp_{psdr} $, they can generate the $n^{th}$ element in this list (step 3 in Figure~\ref{fig:decentralized_poisson_clock}).
    Peers $ u_i $ and $ u_j $ thus can generate equal values of $T_{i, j}^n $, which will be the time between the $n-1^{th}$ and $n^{th}$ ring of the Poisson clock on the edge between $ u_i $ and $ u_j $.
\end{itemize}

By continuously incrementing $ n $, $ Exp_{psdr} $ enables any peer $ u_i $ and $ u_j $ in the network that have access to the same $ Exp_{psdr} $ to generate an infinite, common sequence of clock ringing times. 
The average time period between two sequential rings is $ \alpha$.

Importantly, in the beginning, \sys assumes that the network agrees on a common starting epoch.
This could, for example, be an Unix epoch time and correspond to the first day that a particular application is deployed.
Afterward, peers should periodically invoke the constructed $ Exp_{psdr} $ and keep track of clock rings, relative to the starting epoch, in order to predict when the next clock ring occurs.

\end{document}